\documentclass[aps,reprint,superscriptaddress]{revtex4-2}

\usepackage{graphicx}
\graphicspath{{./figure/}}

\usepackage{amsmath,amssymb,bm,braket}
\usepackage{enumitem}
\usepackage{hyperref}

\usepackage{amsthm}
\usepackage{mdframed}
\newmdtheoremenv{lemma}{Lemma}
\newmdtheoremenv{theorem}{Theorem}
\newmdtheoremenv{result_commutative}{Property (when $\hat{\bm{X}}_N$ is commutative)}
\newmdtheoremenv{result}{Property}

\usepackage{mathtools}
\DeclarePairedDelimiterX{\infdivx}[2]{(}{)}{  #1\;\delimsize\|\;#2}
\newcommand{\infdiv}{D\infdivx}

\newcommand{\appropto}{\mathrel{\vcenter{
  \offinterlineskip\halign{\hfil$##$\cr
    \propto\cr\noalign{\kern2pt}\sim\cr\noalign{\kern-2pt}}}}}

\newcommand{\Tr}{\mathrm{Tr}}
\newcommand{\s}{s_\mathrm{TD}}
\newcommand{\MCE}{MCE}
\newcommand{\CE}{CE}
\newcommand{\indexMCE}{\mathrm{mc}}
\newcommand{\indexCE}{\mathrm{c}}
\newcommand{\ind}[1]{
  \ensuremath{{#1}}}
\newcommand{\indind}[2]{
  \ensuremath{{#1,#2}}}

\usepackage{color}

\begin{document}

\title{Statistical ensembles for phase coexistence states\\ specified by noncommutative additive observables}
\author{Yasushi Yoneta}
\email{yasushi.yoneta@riken.jp}
% \thanks{\\Present address: RIKEN Center for Quantum Computing, 2-1 Hirosawa, Wako City, Saitama 351-0198, Japan}
% \email{yoneta@as.c.u-tokyo.ac.jp}
\affiliation{Department of Basic Science,
The University of Tokyo, 3-8-1 Komaba, Meguro, Tokyo 153-8902, Japan}
\affiliation{RIKEN Center for Quantum Computing,
2-1 Hirosawa, Wako City, Saitama 351-0198, Japan}
\author{Akira Shimizu}
\email{shmz@as.c.u-tokyo.ac.jp}
\affiliation{Institute for Photon Science and Technology,
The University of Tokyo, 7-3-1 Hongo, Bunkyo-ku, Tokyo 113-0033, Japan}
\affiliation{RIKEN Center for Quantum Computing,
2-1 Hirosawa, Wako City, Saitama 351-0198, Japan}
\date{\today}
\begin{abstract}
A phase coexistence state cannot be specified uniquely
by any intensive parameters, such as the temperature and the magnetic field,
because they take the same values over all coexisting phases.
It can be specified uniquely only by an appropriate set of additive observables.
Hence, to analyze phase coexistence states
the statistical ensembles that are specified by additive observables
have been employed,
such as the microcanonical and restricted ensembles.
However, such ensembles are ill-defined or ill-behaved
when some of the additive observables do not commute with each other.
Here, we solve this fundamental problem by extending a generalized ensemble
in such a way that it is applicable to phase coexistence states
which are specified by noncommutative additive observables.
We prove that this ensemble correctly gives
the density matrix corresponding to phase coexistence states of general quantum systems
as well as the thermodynamic functions.
Furthermore, these ensembles are convenient for practical calculations
because of their good analytic properties and useful formulas
by which temperature and other intensive parameters
are directly obtained from the expectation values of the additive observables.
As a demonstration, we apply our formulation to a two-dimensional system
whose phase coexistence states are specified
by an additive observable (order parameter)
that does not commute with the Hamiltonian.
\end{abstract}
\maketitle

\section{Introduction} \label{sec:introduction}
Phase coexistence states are widely observed in nature
and have long been attracting much attention.
For example, the interfaces separating the coexisting phases,
called `phase interfaces' or `domain walls',
play key roles in various fields of science,
such as surface physics, soft matter physics, chemistry, and biology
\cite{Landau1980,Doi2013,Atkins1998,Hyman2014}.
Phase coexistence states also attracts a lot of attention
from a viewpoint of engineering
because the phase interfaces can possess exotic properties
that are absent in the bulk,
such as multiferroicity, superconductivity, and semiconductivity \cite{Salje2009}.
For instance, there is a growing movement
to realize next-generation devices
utilizing such functional properties of phase interfaces
\cite{Parkin2008,Catalan2012,Catalan2014,Sharma2022}.

Recently, the use of quantum fluctuations has been proposed
to control the phase interface at low temperatures,
where thermal fluctuations are negligible,
and the experiment has been conducted \cite{Kagawa2016}.
However, it is a challenging task with current technology
to observe the microscopic structure of the phase interface
because the experiments are performed under extreme conditions.
Therefore, it is important to theoretically analyze quantum systems that exhibit phase coexistence.

However, phase coexistence states are particularly difficult to study theoretically
because they are related to first-order phase transitions
where the thermodynamic properties have particularly strong singularities.
In fact, there has been no concrete method
that can be systematically applied
to phase coexistence states of general quantum systems
for the following reason.

In statistical mechanics, statistical ensembles play a fundamental role
\cite{Gibbs1902,Landau1980,Toda1983,Callen1985},
and various types of ensembles have been devised
\cite{Hetherington1987,Challa1988_PRL,Challa1988_PRA,Johal2003,Ray1991,Gerling1993,Tsallis1988,Beck2003,Cohen2004,Costeniuc2005,Costeniuc2006,Toral2006,Penrose1971,Ellis2000,Touchette2010,Yoneta2019}.
The thermodynamic functions of the macroscopic systems with short-range interactions
are independent of the ensemble used in the calculation \cite{Ruelle1999}.
This fact is called the `equivalence of ensembles'.
Therefore, to obtain thermodynamic functions of the macroscopic system,
one can employ any ensemble.
However, to investigate other properties,
such as microscopic structures of the phase interfaces
or thermodynamics of the systems with finite-size effects,
one must choose an appropriate ensemble.

The appropriate choice of the ensemble is particularly important
in the first-order phase transition region,
where several phases can coexist in various proportions
\cite{Gibbs1876,Gibbs1878}.
The coexisting phases in a phase coexistence state
cannot be distinguished by any intensive parameters,
such as the temperature and the magnetic field,
because every intensive parameter
takes the same value between the coexisting phases.
For this reason, the canonical ensemble ({\CE})
that specified by intensive parameters
can give only particular states
among the various states in the first-order phase transition region.

For example,
when the periodic boundary conditions are imposed,
the {\CE} gives either a single-phase state
or a statistical mixture of single-phase states,
and it cannot give a phase coexistence state
(see Refs.~\cite{Gross2001,Yoneta2019};
we will also give examples in Section~\ref{sec:d2Ising}).
This problem of the {\CE} can be solved for some specific models
by imposing clever boundary conditions \cite{Dobrushin1973,vanBeijeren1975,Landau2014}.
However, for general systems such as systems with a disordered phase,
it is not clear whether such boundary conditions exist
and what boundary conditions should be imposed.

The coexisting phases in a phase coexistence state can be distinguished 
not by an intensive parameter but by a proper additive quantity,
called the `order parameter' (in a broad sense).
Therefore, the microcanonical ensemble ({\MCE})
can give phase coexistence states for general systems
because it is specified by a set of additive quantities
including the order parameter.
This makes it possible to investigate microscopic structures,
such as those of phase interfaces,
of the phase coexistence state \cite{Harris1985}.

The characteristics of different ensembles are also pronounced
when investigating finite systems
because the equivalence of ensembles does not hold for finite systems
\cite{Harris1985,Stump1987,Gross2001}.
For example, in finite systems which undergo first-order phase transitions
in the thermodynamic limit (TDL),
the concavity of the microcanonical entropy is broken generally,
even with short-range interactions
(see Appendix~C of Ref.~\cite{Yoneta2019}).
This concavity breaking leads to thermodynamic anomalies
such as the negative specific heat
\cite{Bixon1989,Labastie1990,Gross1990,Gross1997,Gross2001,Nielsen1994,Jellinek2000,Reyes-Nava2003},
which was indeed observed experimentally
\cite{DAgostino2000,Schmidt2001,Gobet2002}.
Nevertheless, 
the {\CE} fails to give such thermodynamic anomalies of finite systems, 
while the {\MCE} correctly gives them directly and quantitatively
\cite{Harris1985,Stump1987,Gross2001,Junghans2006,Junghans2008,Chen2009,Penrose1971,Binder1980,Troster2012}.

Unfortunately, however,
the {\MCE} has a fundamental problem in quantum systems
when the additive observables
that specify the {\MCE} do not commute with each other.
Since such observables cannot be diagonalized simultaneously, 
it is generally impossible to construct the {\MCE}
using their simultaneous eigenstates.
A similar problem is also present in the restricted ensemble (RE),
which is obtained by projecting the {\CE} onto the subspace spanned by the eigenstates of an additive observable
\cite{Forster1975,Goldenfeld1992,Penrose1971,Ellis2000}.
Such a projection often generates
a superposition of macroscopically distinct states \cite{Tatsuta2018},
which cannot be a thermal equilibrium state.
(See Section~\ref{sec:noncommutativity}
for these problems of the {\MCE} and the RE.)

To resolve the problem of noncommutative additive observables,
von Neumann proposed constructing the {\MCE}
using a commutative set of observables
that approximate the original set of additive observables \cite{Neumann1929}.
However, such a set of commutative observables are hard to construct,
and their explicit forms are obtained
only in limited cases \cite{Davidson1985,Lin1997,Ogata2013,Hastings2009}.
That is, so far, no method has been established
for constructing the {\MCE} or any other ensembles for the  noncommutative additive observables for general quantum systems.

In this paper, we solve this fundamental problem by extending a generalized ensemble
in such a way that it can be applied to phase coexistence states
which are specified by noncommutative additive observables.
Among various generalized ensembles,
we focus on the squeezed ensemble (SE) \cite{Yoneta2019},
which is a particularly broad class of statistical ensembles.
Originally, the SE was introduced for systems
whose equilibrium states can be specified uniquely
only by the internal energy and the number of spins.
We extend the SE so that it can be defined for systems
whose equilibrium states are specified
by three or more additive observables which can be noncommutative.

For this ensemble, we prove that
all the additive observables have macroscopically definite values
even when they do not commute with each other.
Therefore, by using the SE,
one can correctly obtain the desired phase coexistence states
of general quantum systems without any ad hoc procedures,
such as devising boundary conditions.

In addition, we derive useful formulas for thermodynamic quantities,
such as the thermodynamic entropy and the intensive parameters.
The intensive parameters can be calculated efficiently by the formulas,
without resorting to numerical differentiation of the thermodynamic function.
Therefore, by using the SE,
one can fully investigate the thermodynamic properties of the equilibrium states.

This paper is organized as follows.
In Section~\ref{sec:setup-notation},
we describe the setup and notation.
In Section~\ref{sec:noncommutativity},
we review the difficulties faced by conventional ensembles
which have been employed to study phase coexistence states,
due to the noncommutativity of additive observables.
In Section~\ref{sec:SE},
we define the SE for quantum systems
whose equilibrium states are specified by noncommutative additive observables.
In Sections~\ref{sec:SE_commutative} and \ref{sec:SE_noncommutative},
we investigate the basic properties of the SE.
We show that the SE always gives the correct equilibrium state.
In addition, we derive formulas for the thermodynamic entropy density and intensive parameters.
First, in Section~\ref{sec:SE_commutative},
we consider the case where the additive observables commute with each other.
Next, in Section~\ref{sec:SE_noncommutative},
we consider the noncommutative case.
In Section~\ref{sec:SE_parameter},
we discuss the parameter dependence of the SE.
In Section~\ref{sec:SE_usage},
we summarize the usage of the SE.
In Section~\ref{sec:freeSpin},
we apply our formulation to the free spins
and demonstrate its validity.
In Section~\ref{sec:d2Ising},
we apply our formulation to the two-dimensional transverse field Ising model,
which has coexisting phases distinguished by the order parameter
which does not commute with the Hamiltonian,
and demonstrate that the SE successfully gives phase coexistence states
even in such a case which cannot be treated by conventional ensembles.
In Section~\ref{sec:summary}, we provide our conclusions and discussion.

\section{Setup and notation} \label{sec:setup-notation}
In this section, we describe the setup and notation of our formulation.

\subsection{Setup}
We consider a sequence,
indexed by the number of lattice sites $N$,
of quantum spin systems on the hypercubic lattice
with short-range interactions.
We assume that the equilibrium state 
for each $N$ can be specified uniquely
by a set of $m$ additive quantities
$(X_{\ind{0}}, X_{\ind{1}},\cdots,X_{\ind{m-1}})$,
where 
\begin{align}
  X_{\ind{0}} = U
\end{align}
is the internal energy.
Let $(\hat{X}_{\indind{N}{0}},\hat{X}_{\indind{N}{1}},\cdots,\hat{X}_{\indind{N}{m-1}})$ be a set of the corresponding additive observables,
where 
\begin{align}
  \hat{X}_{\indind{N}{0}} = \hat{H}_N 
\end{align}
is the Hamiltonian.
Here, we do not include the interactions with external fields 
in the {\em internal} energy $U$ or the corresponding observable $\hat{H}_N$,
as in the case of the {\MCE}~\footnote{
For the case of the {\MCE},
the interactions with external fields are introduced into the Hamiltonian
through the Legendre-Fenchel transformation when going to the {\CE}.
},
because we choose additive quantities $(X_{\ind{0}},X_{\ind{1}},\cdots,X_{\ind{m-1}})$,
rather than the temperature and external fields, 
as independent thermodynamic parameters.
Hereafter, we assume that all quantities are nondimensionalized
using an appropriate scale.

To take the thermodynamic limit,
we introduce an additive quantity per site 
\begin{align}
  x_{\ind{i}} \equiv X_{\ind{i}}/N
  \qquad (i=0, 1, \cdots, m-1).
\end{align}
It gives the average density in a phase coexisting state, while it is just the density in a homogeneous state, of the additive quantity.
For simplicity, throughout this paper,
we use the term ``additive quantity density''
even when two or more phases coexist.
Particularly, $u \equiv U/N$ is the energy density.
The corresponding observable is given by
\begin{align}
  \hat{x}_{\indind{N}{i}} \equiv \hat{X}_{\indind{N}{i}}/N
  \qquad (i=0, 1, \cdots, m-1),
\end{align}
and we call it ``additive observable density''.
Particularly, we write $\hat{h}_N \equiv \hat{H}_N/N$.

For simplicity of notation,
sets of $m$ physical quantities are denoted by bold symbols like 
\begin{align}
  \bm{x} &= (x_{\ind{0}},x_{\ind{1}},\cdots,x_{\ind{m-1}}),\\
  \hat{\bm{x}} &= (\hat{x}_{\ind{0}},\hat{x}_{\ind{1}},\cdots,\hat{x}_{\ind{m-1}}).
\end{align}
We also define the product of sets of physical quantities,
$\bm{x}$ and $\bm{y}$, as
\begin{align}
  \bm{x}\cdot\bm{y} \equiv \sum_i x_{\ind{i}}y_{\ind{i}}.
\end{align}

The thermodynamic state space $\Omega$ 
is an open subset of $\mathbb{R}^m$, spanned by $\bm{x}$.

\subsection{Thermodynamic entropy density and genuine thermodynamic quantities}
Let $\s$ be the thermodynamic entropy density (in the thermodynamic limit).
Since the thermodynamic entropy is defined purely thermodynamically,
$\s$ exists uniquely regardless of microscopic details of the system
such as the noncommutativity of $\hat{\bm{X}}_N$.

According to thermodynamics,
$\s$ is a function of $\bm{x}$, $\s=\s(\bm{x})$, which is concave in $\Omega$.
Note that $\s$ is {\em strictly} concave
except in first-order phase transition regions.
On the other hand, in first-order phase transition regions
$\s$ becomes linear in a certain direction(s)
and therefore {\em not} strictly concave \cite{Shimizu2021_1_eng,Callen1985}.
This is the reason why the {\CE} cannot give a phase coexistence state \cite{Yoneta2019}.

In addition, $\s$ is continuously differentiable everywhere in $\Omega$
\cite{Callen1985,Lieb1999,Shimizu2021_1_eng}.
We furthermore employ the standard assumption 
for systems with short-range interactions that 
$\s$ is also twice continuously differentiable~\footnote{
The authors do not know systems with short-range interactions which violate this assumption.
}.

We call $\s$ and its derivatives the `genuine thermodynamic quantities'.
In particular, the first derivatives are the (entropic) intensive parameters \cite{Callen1985}.
We use $\Pi_{\ind{i}}$ to denote the intensive parameter conjugate to $X_{\ind{i}}$, i.e.,
\begin{align}
  \Pi_{\ind{i}}(\bm{x}) \equiv \frac{\partial \s}{\partial x_{\ind{i}}}(\bm{x}).\label{eq:def_Pi}
\end{align}
Particularly, $\beta(\bm{x}) \equiv \Pi_{\ind{0}}(\bm{x})$ is the inverse temperature
of the equilibrium state specified by $\bm{x}$.

\subsection{Canonical ensemble} \label{sec:CE}
One way to obtain $\s$ from statistical mechanics is to use the canonical ensemble ({\CE}),
which is defined by
\begin{align}
  \hat{\rho}_N^\indexCE(\bm{\pi})
  &\equiv \frac{\displaystyle e^{-N\bm{\pi}\cdot\hat{\bm{x}}_N}}{\displaystyle \Tr \left[e^{-N\bm{\pi}\cdot\hat{\bm{x}}_N}\right]}, \label{eq:CE-rho_def}\\
  \psi_N^\indexCE(\bm{\pi})
  &\equiv - \frac{1}{N} \log \Tr \left[e^{-N\bm{\pi}\cdot\hat{\bm{x}}_N}\right], \label{eq:CE-psi_def}
\end{align}
where
\begin{align}
  \bm{\pi}=(\pi_{\ind{0}},\pi_{\ind{1}},\cdots,\pi_{\ind{m-1}}) \in \mathbb{R}^m
\end{align}
is the parameter of the {\CE}
and is equal to the set of intensive parameters $\bm{\Pi}$
defined by Eq.~\eqref{eq:def_Pi}.
In the thermodynamics limit, the thermodynamic function
\begin{align}
  \psi^\indexCE(\bm{\pi}) \equiv \lim_{N\to\infty} \psi_N^\indexCE(\bm{\pi})
\end{align}
given by the {\CE} agrees with the Massieu function.
Therefore, one can obtain the thermodynamic entropy density $\s$
by Legendre-Fenchel transforming $\psi^\indexCE$ \cite{Costeniuc2005}:
\begin{align}
  \s(\bm{x}) = \inf_{\bm{\pi}} \{ \bm{\pi}\cdot\bm{x} - \psi^\indexCE(\bm{\pi}) \}. \label{eq:psi-s_CE}
\end{align}
That is, although the {\CE} fails to give the correct {\em equilibrium state}
at the first-order transition region,
it does give the correct {\em thermodynamic function}
everywhere in the thermodynamic state space
in the thermodynamics limit.

\section{Difficulties with noncommutative additive observables} \label{sec:noncommutativity}
In this section, we briefly review the difficulties
faced by conventional ensembles
which have been employed to analyze phase coexistence states,
due to the noncommutativity of additive observables.

As discussed in Section~\ref{sec:introduction},
to analyze phase coexistence states,
the {\MCE} and the RE, which are specified by the values of additive observables,
have conventionally been employed.
However, a straightforward generalization of such ensembles to the noncommutative case
is ill-defined, or lead to pathological behavior of the density matrix in general.

To illustrate this point let us take a simple example of free spins, for which
\begin{align}
  \hat{H}_N = \hat{X}_{\indind{N}{0}} &= \sum_{i=1}^N \hat{\sigma}_i^x, \label{eq:X0_free_spin}\\
  \hat{X}_{\indind{N}{1}} &= \sum_{i=1}^N \hat{\sigma}_i^y. \label{eq:X1_free_spin}
\end{align}
Obviously, $\hat{X}_{\indind{N}{0}}$ and $\hat{X}_{\indind{N}{1}}$ do not commute with each other.
This gives rise to the following difficulties,
and we cannot employ either the {\MCE} or the RE,
which are conventionally used
in the analysis of classical systems with first-order phase transitions.

First, we consider the {\MCE}.
If $\hat{x}_{\indind{N}{0}}$ and $\hat{x}_{\indind{N}{1}}$ were able to be diagonalized simultaneously,
the {\MCE} could be defined as an equally weighted mixture of their simultaneous eigenstates $\ket{x_{\ind{0}},x_{\ind{1}}}$
corresponding to eigenvalues $(x_{\ind{0}},x_{\ind{1}})$
in a given two-dimensional interval $I \subset \mathbb{R}^2$,
called the `shell',
and its density matrix could be given by
\begin{align}
  \hat{\rho}_N^\indexMCE \propto \sum_{(x_{\ind{0}},x_{\ind{1}}) \in I} \ket{x_{\ind{0}},x_{\ind{1}}}\bra{x_{\ind{0}},x_{\ind{1}}}.
\end{align}
However, since $\hat{x}_{\indind{N}{0}}$ and $\hat{x}_{\indind{N}{1}}$ do not commute,
the simultaneous eigenstates do not exist except when $(x_{\ind{0}},x_{\ind{1}})=(0,0)$
(which corresponds to the state with infinite temperature).
Therefore, the density matrix of the {\MCE} at finite temperature is ill-defined and impossible to construct
\footnote{
In Refs.~\cite{Truong1974,Kastner2010_PRL,Kastner2010_JStatMech},
it was proposed to extend the microcanonical entropy (not the density matrix) to quantum systems
whose equilibrium states are specified by two noncommutative additive observables.
}.

Next, we consider the RE.
Its density matrix has the following form \cite{Forster1975}
\begin{align}
  \hat{\rho}_N^\mathrm{r} \propto \hat{P}_N e^{-\pi_{\ind{0}}\hat{X}_{\indind{N}{0}}} \hat{P}_N,
\end{align}
where $\hat{P}_N$ is the projection operator onto the subspace
spanned by the eigenstates of $\hat{x}_{\indind{N}{1}}$
whose eigenvalues are in the range $[m^-,m^+]$.
Here we take $m^\pm$ to be $N$-independent constants
that satisfy $0<m_-<m_+<1$.
Then,
$\hat{x}_{\indind{N}{0}}$ has an anomalously large fluctuation in $\hat{\rho}_N^\mathrm{r}$.
In fact, as proved in Appendix~\ref{sec:derivation_freeSpin_RE},
the variance of $\hat{x}_{\indind{N}{0}}$ in $\hat{\rho}_N^\mathrm{r}$ is $\Theta(N^0)$~\footnote{
We write $f(N)=\Theta(g(N))$ when $f$ is bounded both above and below by $g$ up to a constant factor in the thermodynamic limit.
In addition, to emphasize that we are considering the asymptotic behavior as a function of $N$, we write $g(N)$ as $N^0$ when $g(N)=1$.
}:
\begin{align}
  \lim_{N\to\infty} \Tr \left[ \left( \hat{x}_{\indind{N}{0}} - x_{\indind{N}{0}}^\mathrm{r} \right)^2 \hat{\rho}_N^\mathrm{r} \right]
  \geq m^- (1-m^-) \tanh^2\pi_{\ind{0}}, \label{eq:freeSpin_RE_var}
\end{align}
where $x_{\indind{N}{0}}^\mathrm{r} \equiv \Tr \left[ \hat{x}_{\indind{N}{0}} \hat{\rho}_N^\mathrm{r} \right]$.
However, it should be $o(N^0)$
\footnote{
We write $f(N)=o(g(N))$ when $f$ is dominated by $g$ in the thermodynamic limit
such that $\displaystyle \lim_{N\to\infty} f(N)/g(N) = 0$.
}
if $\hat{\rho}_N^\mathrm{r}$ were a macroscopically definite state,
such as a thermal equilibrium state.
To make matters worse,
$\hat{\rho}_N^\mathrm{r}$ contains superpositions of macroscopically distinct states
with a significant magnitude \cite{Tatsuta2018}.
That is, the state given by the RE is quite far from thermal equilibrium in general.
As discussed in Ref.~\cite{Tatsuta2018},
the above argument holds similarly for interacting systems.

We solve these problems of the conventional ensembles by extending the SE to quantum systems 
whose equilibrium states are specified by a set of noncommutative additive observables.
Although it might look straightforward extension
of the SE that is specified only by the internal energy \cite{Yoneta2019},
its validity as a thermal ensemble is never obvious
as we have seen for the {\MCE} and RE and hence needs to be proved.

For this reason, we will clearly define the SE,  
derive useful formulas, and show their validity, 
for equilibrium states that are specified by a set of
noncommutative additive observables.

\section{Squeezed ensemble} \label{sec:SE}
Before describing detailed formulation,
we explain the basic idea and the definitions of the SE in this section.
We will prove various properties in the subsequent sections, 
whose main results are summarized
in Sections~\ref{sec:summary_commutative}, \ref{sec:summary_noncommutative} and \ref{sec:Summary_parameter}.
We will also summarize the usage of the SE in Section~\ref{sec:SE_usage}.

\subsection{Basic Idea}
When equilibrium states (for each $N$)
are specified only by the internal energy $U$,
the density matrix of the SE is defined as \cite{Yoneta2019}
\begin{align}
  \hat{\rho}_N^\eta (\kappa) \propto e^{-N \eta(\kappa;\hat{h}_N)}.
\end{align}
Here, $\eta(\kappa;\hat{h}_N)$ is the operator 
that is obtained by substituting the Hamiltonian density $\hat{h}_N$
for the internal energy density $u$
in a function $\eta(\kappa;u)$, 
which is parametrized by an ensemble parameter $\kappa$.
Although the parameter $\kappa$ is not a thermodynamic quantity,
it specifies an equilibrium state uniquely \cite{Yoneta2019}.
By imposing several conditions on $\eta(\kappa;u)$, 
such as the convexity, we proved in Ref.~\cite{Yoneta2019} that 
$\hat{\rho}_N^\eta (\kappa)$ always gives the correct equilibrium state
even in a first-order phase transition region.

To extend the SE to equilibrium states specified by noncommutative additive observables,
we extend $\eta(\kappa;u)$ to a function $\eta(\bm{\kappa};\bm{x})$ of a set of the additive quantity densities $\bm{x}$, 
parametrized by a set of ensemble parameters $\bm{\kappa}$.
The conditions on $\eta(\kappa;u)$ are also extended
to those on $\eta(\bm{\kappa};\bm{x})$.
The density matrix is obtained
by substituting the additive observable densities $\hat{\bm{x}}_N$ for $\bm{x}$,
\begin{align}
  \hat{\rho}_N^\eta(\bm{\kappa}) \propto e^{-N\eta(\bm{\kappa};\hat{\bm{x}}_N)}.
\end{align}
In this substitution,
the noncommutativity of $\hat{\bm{x}}_N$ should be treated appropriately.
One way to define an operator function from a general real function
is to use the spectral decomposition.
However, this method cannot be used in our case
because $\hat{\bm{x}}_N$ cannot be diagonalized simultaneously.
Hence, we take $\eta(\bm{\kappa};\bm{x})$ as a polynomial of $\bm{x}$.
(This can be generalized to a well-behaved power series of $\bm{x}$.)
In this case, $\eta(\hat{\bm{x}}_N)$ is given by the sum of products of $\hat{x}_{\indind{N}{i}}$,
and thus it does not require simultaneous eigenstates for its definition.
Since $\eta(\hat{\bm{x}}_N)$ can be different
depending on the order of the products of $\hat{x}_{\indind{N}{i}}$
that appear when calculating $\eta(\hat{\bm{x}}_N)$,
we let $\eta(\bm{\kappa};\bm{x})$ be a polynomial defined up to the order of the product.
(Such a polynomial is generally called a `noncommutative polynomial'.)
Although the order of the product needs to be specified to {\em define} the SE,
it is irrelevant to the {\em results} obtained from the SE 
in the thermodynamic limit (see Section~\ref{sec:SE_noncommutative}).

In addition, corresponding to the increase of the dimension of the thermodynamic state space $\Omega$,
we extend the ensemble parameter
from a single real number $\kappa$ to a set of real numbers
\begin{align}
  \bm{\kappa} = (\kappa_{\ind{0}},\kappa_{\ind{1}},\cdots,\kappa_{\ind{p-1}}).
\end{align}
It is natural to take the number $p$ of the parameters
to be the same as the dimension $m$ of the thermodynamic state space $\Omega$,
but it is not necessarily the same and can be greater than $m$:
\begin{align}
  m \leq p
\end{align}
We write $K$ $(\subset \mathbb{R}^p)$ for the space formed by the parameter $\bm{\kappa}$.

\subsection{Definition of the SE} \label{sec:SE_def}
Based on the above idea,
we introduce the SE as a general class of ensembles
which give equilibrium values of local observables, additive observables,
and the genuine thermodynamic quantities
even when equilibrium states are specified by noncommutative additive observables.

Let $\eta(\bm{\kappa};\bm{x})$ be a noncommutative polynomial
(or well-behaved power series~\footnote{
All arguments in Sections~\ref{sec:SE_commutative}-\ref{sec:SE_parameter}
are valid even if $\eta$ is a power series and not a polynomial,
as long as it defines an operator Lipschitz function in the operator norm.
})
with real coefficients in $m$ noncommutative variables $\bm{x}$,
parametrized by $p$ real numbers $\bm{\kappa}$ in $K$.
We define the {\em squeezed ensemble (SE) associated with $\eta$} by
\begin{align}
  \hat{\rho}_N^\eta(\bm{\kappa})
  &\equiv \frac{\displaystyle e^{-N\eta(\bm{\kappa};\hat{\bm{x}}_N)}}{\displaystyle \Tr\left[e^{-N\eta(\bm{\kappa};\hat{\bm{x}}_N)}\right]}, \label{eq:rho_N^eta}
\end{align}
and write the normalization constant as
\begin{align}
  \psi_N^\eta (\bm{\kappa})
  &\equiv - \frac{1}{N} \log \Tr \left[ e^{-N\eta(\bm{\kappa};\hat{\bm{x}}_N)} \right].
  \label{eq:psi_N^eta}
\end{align}
Note that simultaneous eigenstates of $\hat{\bm{X}}_N$
do not appear in the definition of the SE.
Therefore, the SE is well-defined
even in the case where $\hat{\bm{X}}_N$ cannot be diagonalized simultaneously.

Since $\eta$ is defined for noncommutative variables,
it is possible to substitute the commutative quantities
for the variables as a special case.
Therefore, the polynomial $\eta$ defines a real function on $\mathbb{R}^m$ uniquely.
When there is no danger of confusion,
we will use the same symbol $\eta$ for this function.
Using this notation,
we impose the following four conditions on $\eta$.
\begin{enumerate}[label={(\Alph*)}]
  \setcounter{enumi}{0}
  \item \label{cond:A}
    $\eta(\bm{\kappa};\hat{\bm{x}}_N)$ is self-adjoint for all $\bm{\kappa} \in K$ and $N$.
\end{enumerate}
\begin{enumerate}[label={(\Alph*)}]
  \setcounter{enumi}{1}
  \item \label{cond:B}
    $\s(\bm{x})-\eta(\bm{\kappa};\bm{x})$ has the unique maximum point
    $\bm{x}_\mathrm{max}^\eta$ in $\Omega$
    and is strongly concave
    in a neighborhood of $\bm{x}_\mathrm{max}^\eta$
    for all $\bm{\kappa} \in K$.
\end{enumerate}
\begin{enumerate}[label={(\Alph*)}]
  \setcounter{enumi}{2}
  \item \label{cond:C}
    $\eta(\bm{\kappa};\bm{x})$ is twice continuously differentiable
    as a $(p+m)$-variable function of $\bm{\kappa}$ and $\bm{x}$.
\end{enumerate}
\begin{enumerate}[label={(\Alph*)}]
  \setcounter{enumi}{3}
  \item \label{cond:D}
    $\bm{x}_\mathrm{max}^\eta(\bm{\kappa})$ is surjective onto $\Omega$. 
\end{enumerate}

Before going into details,
we briefly explain physical meanings of these conditions.
The condition~\ref{cond:A} ensures that $\hat{\rho}_N^\eta$ is a density matrix.
Since all additive observable densities $\hat{x}_{\indind{N}{i}}$ are self-adjoint,
this condition can be easily satisfied by employing an appropriately symmetrized polynomial as $\eta$.
The condition~\ref{cond:B} ensures that
all additive observables that specify the equilibrium state
have macroscopically definite value in the SE.
By contrast, as discussed in Section~\ref{sec:introduction},
at the first-order phase transition point,
the {\CE} gives a statistical mixture of macroscopically distinct states in many cases.
This condition makes the SEs free from such deficiency.
Since $\s$ is concave, this condition is always satisfied 
when $\eta$ is taken as a strongly convex function.
The condition~\ref{cond:C} ensures
that the equilibrium state described by the SE changes continuously
with respect to the parameter $\bm{\kappa}$.
The condition~\ref{cond:D} plays a crucial role for obtaining all thermodynamic properties
from the thermodynamic function given by the SE.

\section{Basic properties when \texorpdfstring{$\hat{\bm{X}}_N$}{XN} is commutative} \label{sec:SE_commutative}
In this and the next section,
we investigate the basic properties of the SE
when the parameter $\bm{\kappa}$ is fixed to
an arbitrary set of values,
and hence we simplify notation
by abbreviating $\eta(\bm{\kappa};\bm{x})$ as $\eta(\bm{x})$
except for the final results (which are displayed in frames).

We first discuss the simple case
where $\hat{\bm{x}}_N$ commute with each other
in this section.
The noncommutative case will be discussed in the next section.

\subsection{Basic properties of density matrix of the SE} \label{sec:densityMatrix}
We examine the probability distribution $p_N^\eta(\bm{x})$
of the additive quantity densities $\bm{x}$ in
$\hat{\rho}_N^\eta$.

Since we are considering the commutative case in this section,
$\hat{\bm{x}}_N$ can be diagonalized simultaneously,
and we can thereby define the density of microstates.
We divide the thermodynamic state space $\Omega$ (more precisely, $\mathbb{R}^m$)
into a direct sum of $m$-dimensional intervals $I_{\bm{n}}$
of side lengths of $o(N^0)$~\footnote{
More specifically, we define $m$-dimensional interval $I_{\bm{n}}$ as
$I_{\bm{n}} \equiv (n_{\ind{0}}\delta_{\ind{0}},(n_{\ind{0}}+1)\delta_{\ind{0}}] \times (n_{\ind{1}}\delta_{\ind{1}},(n_{\ind{1}}+1)\delta_{\ind{1}}] \times \cdots\times (n_{\ind{m-1}}\delta_{\ind{m-1}},(n_{\ind{m-1}}+1)\delta_{\ind{m-1}}]$,
where $\bm{n}=(n_{\ind{0}},n_{\ind{1}},\cdots,n_{\ind{m-1}})\in\mathbb{Z}^m$
is the labeling of the interval
and $\delta_i=o(N^0)$ is the side length of the interval.
Then this gives the direct sum decomposition of $\mathbb{R}^m$.
}
and write $g_N(\bm{x})$ for the density of microstates
averaged within the interval $I_{\bm{n}}$ that contains $\bm{x} \in \Omega$.
By Boltzmann's entropy formula,
the logarithm of $g_N$ approaches the thermodynamic entropy $\s$
in the thermodynamic limit (TDL):
\begin{align}
  \sigma_N(\bm{x})
  \equiv \frac{1}{N} \log g_N(\bm{x})
  \overset{\mathrm{TDL}}{\to} \s(\bm{x}).
  \label{eq:sigma_to_std}
\end{align}
From this relation,
for any smooth function $f(\bm{x})$ we have
\begin{align}
  \Tr \left[ f(\hat{\bm{x}}_N) e^{-N\eta(\hat{\bm{x}}_N)} \right]
  &= \int d\bm{x} g_N(\bm{x}) f(\bm{x}) e^{-N [\eta(\bm{x}) + o(N^0)]} \nonumber\\
  &= \int d\bm{x} f(\bm{x}) e^{N[\sigma_N(\bm{x}) - \eta(\bm{x}) + o(N^0)]} \nonumber\\
  &= \int d\bm{x} f(\bm{x}) e^{N[\s(\bm{x}) - \eta(\bm{x}) + o(N^0)]}.
  \label{eq:tr-int}
\end{align}
Therefore, we obtain
\begin{align}
  p_N^\eta(\bm{x})
  \propto e^{N[\s(\bm{x}) - \eta(\bm{x}) + o(N^0)]}. \label{eq:prob}
\end{align}
From condition~\ref{cond:B},
the function $\s(\bm{x}) - \eta(\bm{x})$ in the exponential
has the unique maximum at $\bm{x}_\mathrm{max}^\eta$, which satisfies
\begin{align}
  \frac{\partial \s}{\partial x_{\ind{i}}}(\bm{x}_\mathrm{max}^\eta)
  = \frac{\partial \eta}{\partial x_{\ind{i}}}(\bm{x}_\mathrm{max}^\eta). \label{eq:x_max^eta}
\end{align}
We investigate properties of $p_N^\eta(\bm{x})$
in the vicinity of this maximum point.
Expanding $\s(\bm{x})-\eta(\bm{x})$ around $\bm{x}_\mathrm{max}^\eta$, we get
\begin{align}
  &\s(\bm{x}) - \eta(\bm{x}) \nonumber\\
  &= \s(\bm{x}_\mathrm{max}^\eta) - \eta(\bm{x}_\mathrm{max}^\eta) \nonumber\\
  &- \frac{1}{2} \sum_{i,j} (x_{\ind{i}}-x_{\indind{\mathrm{max}}{i}}^\eta) (-H_{ij}) (x_{\ind{j}}-x_{\indind{\mathrm{max}}{j}}^\eta)
  + \cdots, \label{eq:taylor}
\end{align}
where $H$ is the Hesse matrix of $\s(\bm{x})-\eta(\bm{x})$ at $\bm{x}_\mathrm{max}^\eta$:
\begin{align}
  H_{ij}
  = \frac{\partial^2 \s}{\partial x_{\ind{i}} \partial x_{\ind{j}}}(\bm{x}_\mathrm{max}^\eta)
  - \frac{\partial^2 \eta}{\partial x_{\ind{i}} \partial x_{\ind{j}}}(\bm{x}_\mathrm{max}^\eta).
\end{align}
Again from condition~\ref{cond:B}, the matrix $H$ is negative definite.
Hence, $e^{N[\s(\bm{x}) - \eta(\bm{x})]}$
behaves as the Gaussian distribution
in the vicinity of $\bm{x}_\mathrm{max}^\eta$,
peaking at $\bm{x}_\mathrm{max}^\eta$,
with the covariance matrix $-N^{-1} H^{-1}$
\begin{align}
  \exp \left[
    - \frac{1}{2} \sum_{i,j} (x_{\ind{i}}-x_{\indind{\mathrm{max}}{i}}^\eta) (- N H_{ij}) (x_{\ind{j}}-x_{\indind{\mathrm{max}}{j}}^\eta)
  \right]. \label{eq:p_N^eta_Gaussian}
\end{align}
Although higher-order terms are dropped in this equation
(Laplace's approximation \cite{Butler2007}),
their contributions become negligible in the thermodynamic limit
because we have taken $\s(\bm{x})-\eta(\bm{x})$ strongly concave.
Taking also the $o(N^0)$ term in $p_N^\eta(\bm{x})$ into account,
we obtain
\begin{align}
  \Tr \left[ f(\hat{\bm{x}}_N) \hat{\rho}_N^\eta \right]
  = f(\bm{x}_\mathrm{max}^\eta) + o(N^0). \label{eq:tr_f-rho_x}
\end{align}

Letting $f(\bm{x})=x_{\ind{i}}$ in Eq.~\eqref{eq:tr_f-rho_x},
we obtain the expectation value of $\hat{\bm{x}}_N$ as
\begin{result_commutative} \label{property:x_N^eta_commutative}
\begin{align}
  x_{\indind{N}{i}}^\eta(\bm{\kappa}) \equiv \Tr \left[ \hat{x}_{\indind{N}{i}} \hat{\rho}_N^\eta(\bm{\kappa}) \right] \overset{\mathrm{TDL}}{\to} x_{\indind{\mathrm{max}}{i}}^\eta(\bm{\kappa}) \label{eq:x_N^eta_commutative}
\end{align}
\end{result_commutative}
Furthermore, letting $f(\bm{x})=(x_{\ind{i}})^2$ in Eq.~\eqref{eq:tr_f-rho_x}
and combining it with Eq.~\eqref{eq:x_N^eta_commutative},
we obtain the variance as
\begin{result_commutative} \label{property:var_x_N^eta_commutative}
\begin{align}
  \lim_{N\to\infty} \Tr \left[ \left( \hat{x}_{\indind{N}{i}} - x_{\indind{N}{i}}^\eta(\bm{\kappa}) \right)^2 \hat{\rho}_N^\eta(\bm{\kappa}) \right] = 0 \label{eq:var_x_N^eta_commutative}
\end{align}
\end{result_commutative}
That is, in the density matrix given by the SE,
all $\hat{x}_{\indind{N}{i}}$ have macroscopically definite values.
This is in contrast to the {\CE},
in which some $\hat{x}_{\indind{N}{i}}$ has macroscopically large fluctuation
at the first-order phase transition point.

Moreover, $\hat{\rho}_N^\eta$ can be characterized by the principle of equal probability.
That is, under the constraint of
\begin{align}
  \mathrm{Tr} \left[ \eta(\hat{\bm{x}}_N) \hat{\rho} \right]
  = \mathrm{Tr} \left[ \eta(\hat{\bm{x}}_N) \hat{\rho}_N^\eta \right],
\end{align}
we can show that the
density matrix $\hat{\rho}$ that maximizes the von Neumann entropy,
\begin{align}
  S^\mathrm{vN}(\hat{\rho}) \equiv - \mathrm{Tr} \left[ \hat{\rho} \log \hat{\rho} \right],
\end{align}
exists uniquely and is equal to $\hat{\rho}_N^\eta$,
from argument similar to that of Refs.~\cite{Jaynes1957_1,Jaynes1957_2,Jaynes1963}.
Note that this argument can also be applied to the case
where $\hat{X}_N$ do not commute with each other.

From these results, we conclude that
the density matrix $\hat{\rho}_N^\eta$ given by the SE
gives typical properties of the microstates
in which additive quantity densities $\bm{x}$ are $\bm{x}_\mathrm{max}^\eta$.
Therefore, by choosing $\eta$ such that $\bm{x}_\mathrm{max}^\eta$
coincides with $\bm{x}$ in the particular equilibrium state of interest,
one can obtain the quantum state corresponding to that equilibrium state.
In particular, by choosing $\eta$
such that $\bm{x}_\mathrm{max}^\eta$ lies
within the first-order phase transition region,
one can obtain a phase coexistence state
without finding and imposing clever boundary conditions
according to the phases to be coexisted
as in the case of using the {\CE}.
Thus, by using the SE,
one can investigate microscopic structures of phase coexisting states,
such as the phase interfaces.

\subsection{Genuine thermodynamic quantities} \label{sec:genuineTDQ}
We can also obtain genuine thermodynamic quantities,
such as the thermodynamic entropy density and the intensive parameter, easily from the SE.

First, we derive a formula for the thermodynamic entropy density.
We note that only $\bm{x}$ which is in close vicinity around the peak position $\bm{x}_\mathrm{max}^\eta$
contributes to the integral in the right hand side of Eq.~\eqref{eq:tr-int}.
Hence, letting $f=1$ and performing the Gaussian integral, we have
\begin{align}
  &\Tr \left[ e^{-N\eta(\hat{\bm{x}}_N)} \right] \nonumber\\
  &= \frac{(2\pi)^{m/2}}{N^{m/2} |H|^{1/2}} e^{N [\s(\bm{x}_\mathrm{max}^\eta)-\eta(\bm{x}_\mathrm{max}^\eta) + o(N^0)]} \left\{1+o(N^0)\right\}.
\end{align}
Here $|H|$ is the determinant of the Hesse matrix $H$.
Then we obtain
\begin{align}
  \psi_N^\eta
  &= - \frac{1}{N} \log \Tr \left[ e^{-N\eta(\hat{\bm{x}}_N)} \right] \nonumber\\
  &= \eta(\bm{x}_\mathrm{max}^\eta) - \s(\bm{x}_\mathrm{max}^\eta) + o(N^0).
\end{align}
Therefore, we obtain
\begin{result_commutative} \label{formula:entropy_commutative}
\begin{align}
  s_N^\eta(\bm{\kappa})
  &\equiv \eta(\bm{\kappa};\bm{x}_N^\eta(\bm{\kappa})) - \psi_N^\eta(\bm{\kappa})
  \overset{\mathrm{TDL}}{\to} \s(\bm{x}_\mathrm{max}^\eta(\bm{\kappa})) \label{eq:entropy-formula_commutative}
\end{align}
\end{result_commutative}
Since $\eta$ is a known function,
one can obtain the thermodynamic entropy density
from $\bm{x}_N^\eta$ and $\psi_N^\eta$ using this formula.

Next, we derive a formula for intensive parameters.
Using Eq.~\eqref{eq:x_max^eta} and in the same manner as above,
we also obtain
\begin{result_commutative} \label{formula:tdforce_commutative}
\begin{align}
  \Pi_{\indind{N}{i}}^\eta(\bm{\kappa})
  &\equiv \frac{\partial\eta}{\partial x_{\ind{i}}}(\bm{\kappa};\bm{x}_N^\eta)
  \overset{\mathrm{TDL}}{\to} \Pi_{\ind{i}}(\bm{x}_\mathrm{max}^\eta(\bm{\kappa})) \label{eq:tdforce-formula_commutative}
\end{align}
\end{result_commutative}
Since $\eta$ is a known function,
using this formula,
one can obtain the intensive parameters
just by calculating $\bm{x}_N^\eta$.
This is a great advantage of the SE over the {\MCE}.
Since we consider the commutative case in this section,
the {\MCE} is well defined (unlike the noncommutative case).
However,
one needs to differentiate the entropy
in order to calculate the intensive parameters from the {\MCE},
and it gives very noisy results in numerical calculation.
One can avoid such noisy calculation
by using formula Eq.~\eqref{eq:tdforce-formula_commutative} of the SE.

\subsection{Summary of properties when \texorpdfstring{$\hat{\bm{X}}_N$}{XN} is commutative}
\label{sec:summary_commutative}
We have discussed the case where $\hat{\bm{X}}_N$ is commutative in this section.
By using the SE, one can correctly obtain the desired phase coexistence states
without any ad hoc procedures, such as devising boundary conditions.
One can also obtain genuine thermodynamic quantities,
such as the thermodynamic entropy density and the intensive parameter, easily from the SE.
Therefore, one can investigate microscopic structures and thermodynamic properties of phase coexistence states.

In the next section, we will show that
all these properties are kept even in the noncommutative case.

\section{Basic properties when \texorpdfstring{$\hat{\bm{X}}_N$}{XN} is noncommutative}  \label{sec:SE_noncommutative}
We now study the noncommutative case.
As in the previous section,
we fix $\bm{\kappa}$ to an arbitrary set of values
and abbreviate $\eta(\bm{\kappa};\bm{x})$ as $\eta(\bm{x})$
except for the final results.

\subsection{Strategy for proofs and derivations} \label{sec:SE_noncommutative_strategy}
When $\hat{\bm{X}}_N$ do not commute with each other,
the density of microstates $g_N(\bm{x})$ is ill-defined.
Hence, we cannot apply the arguments,
such as Laplace's approximation,
of the previous section directly.

To overcome this obstacle,
we bring the idea proposed by von Neumann of using a commutative set of observables $\check{\bm{x}}_N$
that approximate the set of additive observable densities $\hat{\bm{x}}_N$.
Although such a set of commutative observables are hard to construct,
we use it only for the proofs and derivations of the basic properties of the SE,
assuming only the existence of $\check{\bm{X}}_N$.

After establishing the validity and formulas of the SE,
we can use the set of original additive observable densities $\hat{\bm{x}}_N$
when applying the SE to concrete calculations.
That is, we can calculate statistical-mechanical quantities
without using $\check{\bm{x}}_N$,
which would be practically impossible to construct.

\subsection{Commutative set of observables that approximates \texorpdfstring{$\hat{\bm{X}}_N$}{XN}} \label{sec:SE_noncommutative_hat-check}
Let us introduce the commutative set of observables $\check{\bm{x}}_N$.
Since we are considering a general quantum spin system,
its local Hilbert space on each site is taken as $\mathbb{C}^d$.
The spins are located on the $\nu$-dimensional hypercubic lattice,
denoted by $\Lambda_n=[-n,+n]^\nu \cap \mathbb{Z}^\nu$.
Then, an additive observable $\hat{X}_{\indind{N}{i}}$ is expressed as~\footnote{
Even when $\hat{X}_{\indind{N}{i}}$ contains a boundary terms,
there exists $\check{\bm{x}}_N$ that satisfies Eqs.~\eqref{eq:hat-check} and \eqref{eq:check},
because the boundary terms change the operator norm of $\hat{x}_{\indind{N}{i}}$ by only $O(N^{-1/\nu})$.
}
\begin{align}
  \hat{X}_{\indind{N}{i}} = \sum_{j\in\mathbb{Z}^\nu \text{\ s.t.\ }\gamma_j(I_{\ind{i}}) \subset \Lambda_n} \gamma_j(\hat{o}_{\ind{i}}), \label{eq:X_def}
\end{align}
where $\gamma_j$ is the $j$-lattice translation for $j\in\mathbb{Z}$
and $\hat{o}_{\ind{i}}$ is an $N$-independent local observable with support $I_{\ind{i}}$.
For this system, Ogata~\cite{Ogata2013} proved that
there exists a set of $m$ observables
\begin{align}
  \check{\bm{x}}_N = (\check{x}_{\indind{N}{0}},\check{x}_{\indind{N}{1}},\cdots,\check{x}_{\indind{N}{m-1}})
\end{align}
such that
\begin{align}
  \lim_{N\to\infty} \left\| \hat{x}_{\indind{N}{i}}-\check{x}_{\indind{N}{i}} \right\| &= 0
  \quad (i=0, 1, \cdots, m-1), \label{eq:hat-check}\\
  \left[ \check{x}_{\indind{N}{i}}, \check{x}_{\indind{N}{j}} \right] &= 0
  \quad (i,j=0, 1, \cdots, m-1). \label{eq:check}
\end{align}
To put it differently,
there exists a commutative set of observables $\check{\bm{x}}_N$
that approximates a noncommutative set of additive observable densities $\hat{\bm{x}}_N$
in the thermodynamic limit.

For notational simplicity,
we use the notation $\dot{\bullet}$
to represent the quantity $\hat{\bullet}$ associated with $\hat{\bm{x}}_N$
and the quantity $\check{\bullet}$ associated with $\check{\bm{x}}_N$ together.
Using this notation, we define two types of SEs,
one for $\hat{\bm{x}}_N$ and the other for $\check{\bm{x}}_N$, as
\begin{align}
  \rho_N^{\dot{\eta}}
  &\equiv \frac{\displaystyle
    e^{-N\eta(\dot{\bm{x}}_N)}
  }{\displaystyle
    \Tr \left[ e^{-N\eta(\dot{\bm{x}}_N)} \right]
  },\\
  \psi_N^{\dot{\eta}}
  &\equiv - \frac{1}{N} \log \Tr \left[ e^{-N\eta(\dot{\bm{x}}_N)} \right],
\end{align}
where $\dot{\bm{x}}_N = \hat{\bm{x}}_N$ or $\check{\bm{x}}_N$.
For the expectation value, we write
\begin{align}
  \braket{\bullet}_N^{\dot{\eta}}
  &\equiv \Tr \left[ \bullet \rho_N^{\dot{\eta}} \right].
\end{align}
We prove the following two theorems
in Appendices~\ref{sec:proof_theorem_psi} and \ref{sec:proof_theorem_exp}.

First, the thermodynamic function given by the SE for $\hat{\bm{x}}_N$
coincides with that given by the SE for $\check{\bm{x}}_N$:
\begin{theorem} \label{theorem:psi}
If $\psi_N^{\hat{\eta}}$ and $\psi_N^{\check{\eta}}$ converge in the thermodynamic limit, then
\begin{align}
  \lim_{N\to\infty} \psi_N^{\hat{\eta}}
  = \lim_{N\to\infty} \psi_N^{\check{\eta}}.
\end{align}
This holds even when condition~\ref{cond:B} is not satisfied.
\end{theorem}
In the proof we use only the fact that $\eta$ satisfies condition~\ref{cond:A}
and is Lipschitz continuous in the operator norm.
Since condition~\ref{cond:B} is unnecessary,
this theorem holds
even for $\eta=\bm{\pi}\cdot\bm{x}$,
which associates the {\CE} (see Section~\ref{sec:parameter_CE})
and does not satisfy condition~\ref{cond:B}
at the phase transition point.
This fact will be crucial
when deriving formulas for the genuine thermodynamic quantities
in Section~\ref{sec:SE_noncommutative_hat}.

Next, we consider the expectation values of observables.
Here, we focus on observables,
especially those that can be expressed as polynomials
(denoted as $\theta$)
of the observables $\dot{\bm{x}}$.
This class of observables includes not only $\dot{x}_{\indind{N}{i}}$,
but also their higher-order moments.
We obtain the following theorem:
\begin{theorem} \label{theorem:exp}
Let $\theta$ be a polynomial with real coefficients
in $m$ noncommutative variables.
We assume that $\theta$ is independent of $N$
and that $\theta(\hat{\bm{x}}_N)$ is self-adjoint for all $N$.

If the following conditions are fulfilled:

\begin{enumerate}[label={(\roman*)}]
  \item \label{cond:exp_converge}
    $\left\{\braket{\theta(\hat{\bm{x}}_N)}_N^{\hat{\eta}}\right\}_{N\in\mathbb{N}}$ and $\left\{\braket{\theta(\check{\bm{x}}_N)}_N^{\check{\eta}}\right\}_{N\in\mathbb{N}}$ converge in the thermodynamic limit,
  \item \label{cond:variance}
    There exists a closed and bounded interval $I$ containing $0$ and $M>0$
    such that, for any $\lambda \in I$ and $N\in\mathbb{N}$,
    the variance of $\theta(\check{\bm{x}}_N)$
    in the SE associated with $\eta+\lambda\theta$ for $\check{\bm{x}}_N$ satisfies
    \begin{align}
      \Tr \left[ \left( \theta(\check{\bm{x}}_N) - \braket{\theta(\check{\bm{x}}_N)}_N^{\check{\eta},\check{\theta}}(\lambda) \right)^2 \rho_N^{\check{\eta},\check{\theta}}(\lambda) \right]
      < \frac{M}{N}, \label{eq:variance_cond}
    \end{align}
    where
    \begin{align}
      \rho_N^{\dot{\eta},\dot{\theta}}(\lambda)
      &\equiv \frac{\displaystyle
        e^{-N\left(\eta(\dot{\bm{x}}_N)+\lambda\theta(\dot{\bm{x}}_N)\right)}
      }{\displaystyle
        \Tr \left[ e^{-N\left(\eta(\dot{\bm{x}}_N)+\lambda\theta(\dot{\bm{x}}_N)\right)} \right]
      },\\
      \braket{\bullet}_N^{\dot{\eta},\dot{\theta}}(\lambda)
      &\equiv \Tr \left[ \bullet \rho_N^{\dot{\eta},\dot{\theta}}(\lambda) \right],
    \end{align}
\end{enumerate}
then we have
\begin{align}
  \lim_{N\to\infty} \braket{\theta(\hat{\bm{x}}_N)}_N^{\hat{\eta}}
  = \lim_{N\to\infty} \braket{\theta(\check{\bm{x}}_N)}_N^{\check{\eta}}.
\end{align}
\end{theorem}
These two theorems show that
the statistical-mechanical properties of the SE are identical
for $\hat{\bm{x}}_N$ and for $\check{\bm{x}}_N$ in the thermodynamic limit.

Since condition~\ref{cond:variance} is on the SE
for the commutative observables $\check{\bm{x}}_N$,
it can be shown to hold using Laplace's approximation
under condition~\ref{cond:B} and assumptions listed in Section~\ref{sec:SE_noncommutative_assumption},
as we do in Appendix~\ref{sec:properties_proof}.

\subsection{Assumptions} \label{sec:SE_noncommutative_assumption}
Since $\check{\bm{x}}_N$ can be diagonalized simultaneously,
we can define the density of microstates
that have the specified value of $\check{\bm{x}}_N$.
We write $\check{g}_N$ for the density of microstates averaged
in the same way as we defined $g_N$ in Section~\ref{sec:densityMatrix}.
It is expected that $\check{g}_N(\bm{x})$ has the same properties as 
the density of microstates of the commutative case, $g_N(\bm{x})$.
We therefore make the following reasonable assumptions:
\begin{enumerate}  \setcounter{enumi}{0}
  \item \label{assumption:noncommutativity_concave}
    There exists a concave function $\varsigma(\bm{x})$ such that
    \begin{align}
      \varsigma_N(\bm{x}) 
      \equiv \frac{1}{N} \log \check{g}_N(\bm{x}) \overset{\mathrm{TDL}}{\to} \varsigma(\bm{x}).
    \end{align}
  \item \label{assumption:noncommutativity_Laplace}
    Laplace's approximation,
    which will be applied in Section~\ref{sec:SE_noncommutative_check}
    in the same way as done in Section~\ref{sec:SE_commutative},
    for the SE constructed from $\check{\bm{x}}_N$
    gives exact results in the thermodynamic limit.
\end{enumerate}

Let us discuss the validity of these assumptions.
As pointed out by von Neumann in Ref.~\cite{Neumann1929},
in macroscopic experiments,
it is possible to measure a set of macroscopic quantities simultaneously,
and the observables measured are not $\hat{\bm{x}}_N$
but a set of ``coarse-grained'' commutative observables.
The microcanonical entropy density
for such a set of actually measured macroscopic observables
should be $\varsigma_N$.
Therefore, it is expected that $\varsigma_N$
asymptotically approaches the thermodynamic entropy density $\s$,
which is concave, in the thermodynamic limit.
In fact, we will show 
in Section~\ref{sec:SE_noncommutative_check}
that $\varsigma=\s$ under assumptions \ref{assumption:noncommutativity_concave} and \ref{assumption:noncommutativity_Laplace}.
Thus assumption~\ref{assumption:noncommutativity_concave} is reasonable.

Assumption~\ref{assumption:noncommutativity_Laplace} also seems to hold
unless $\varsigma_N$ exhibits very pathological behavior.
Under assumption~\ref{assumption:noncommutativity_concave}, 
$\varsigma_N(\bm{x})-\eta(\bm{x})$ 
converges to $\varsigma(\bm{x})-\eta(\bm{x})$,
which is strongly concave because of condition~\ref{cond:B}.
Hence, it is expected that $\varsigma_N(\bm{x})-\eta(\bm{x})$ is also strongly concave,
at least in the very vicinity of the peak position
unless $\varsigma_N$ exhibits pathological behavior in the region
(such as having a point at which convergence to the thermodynamic limit is extremely slow).
This corresponds to the fact that 
$\sigma_N(\bm{x})-\eta(\bm{x})$ can be well approximated
by the strongly concave function $\s(\bm{x})-\eta(\bm{x})$
in the commutative case, discussed in Section~\ref{sec:densityMatrix}, 
in which Laplace's approximation gives exact results
in the thermodynamic limit.
Therefore, we expect that 
Laplace's approximation also works well in the case of $\varsigma_N(\bm{x})-\eta(\bm{x})$.

\subsection{Properties of the SE for commutative set \texorpdfstring{$\check{\bm{X}}_N$}{XN}} \label{sec:SE_noncommutative_check}
Following the strategy explained in Section~\ref{sec:SE_noncommutative_strategy}, 
we first study the basic properties of the SE for a commutative set $\check{\bm{x}}_N$.

In preparation, we first discuss the relation between $\varsigma$ and $\s$.
Using $\varsigma_N$,
the thermodynamic function given by the {\CE} for $\check{\bm{x}}_N$ can be expressed as
\begin{align}
  \psi_N^{\check{\indexCE}}(\bm{\pi})
  &\equiv - \frac{1}{N} \log \Tr \left[e^{-N\bm{\pi}\cdot\check{\bm{x}}_N}\right] \nonumber\\
  &= - \frac{1}{N} \log \int d\bm{x} e^{-N[\bm{\pi}\cdot\bm{x}-\varsigma_N(\bm{x})]} + o(N^0).
\end{align}
Since $e^{-N[\bm{\pi}\cdot\bm{x}-\varsigma_N(\bm{x})]}$ has its maximum value
at the point where $\bm{\pi}\cdot\bm{x} - \varsigma_N(\bm{x})$ is minimum
and decreases exponentially away from the point, we find
\begin{align}
  \psi_N^{\check{\indexCE}}(\bm{\pi})
  = \inf_{\bm{x} \in \Omega} \left\{ \bm{\pi}\cdot\bm{x} - \varsigma_N(\bm{x}) \right\} + o(N^0).
\end{align}
Therefore, in the thermodynamic limit,
by using Theorem~\ref{theorem:psi} for $\eta=\bm{\pi}\cdot\bm{x}$,
we obtain
\begin{align}
  \psi^{\indexCE}(\bm{\pi})
  &= \lim_{N\to\infty} \psi_N^{\hat{\indexCE}}(\bm{\pi}) \nonumber\\
  &= \lim_{N\to\infty} \psi_N^{\check{\indexCE}}(\bm{\pi}) \nonumber\\
  &= \inf_{\bm{x} \in \Omega} \left\{ \bm{\pi}\cdot\bm{x} - \varsigma(\bm{x}) \right\}.
\end{align}
Therefore, by Legendre-Fenchel transforming \cite{Costeniuc2005}
both sides and using Eq.~\eqref{eq:psi-s_CE},
we get
\begin{align}
  \varsigma(\bm{x}) = \s(\bm{x}). \label{eq:varsigma_sTD}
\end{align}

Now we examine the probability distribution $p_N^\eta(\bm{x})$ 
of the commutative set $\check{\bm{x}}_N$
in the density matrix $\rho_N^{\check{\eta}}$.
Using Eq.~\eqref{eq:varsigma_sTD},
for any smooth function $f(\bm{x})$, we have
\begin{align}
  \Tr \left[ f(\check{\bm{x}}_N) e^{-N\eta(\check{\bm{x}}_N)} \right]
  &= \int d\bm{x} f(\bm{x}) e^{N [\varsigma_N(\bm{x}) - \eta(\bm{x}) + o(N^0)]} \nonumber\\
  &= \int d\bm{x} f(\bm{x}) e^{N [\s(\bm{x}) - \eta(\bm{x}) + o(N^0)]}. \label{eq:tr-int_check}
\end{align}
This is just Eq.~\eqref{eq:tr-int}.
Therefore, the results for $p_N^\eta$
derived in Section~\ref{sec:SE_commutative} (Eqs.~\eqref{eq:prob}-\eqref{eq:p_N^eta_Gaussian})
hold for $p_N^\eta$ in the SE for $\check{\bm{x}}_N$ as well.
Then, evaluating the statistical-mechanical quantities in the SE for $\check{\bm{x}}_N$, we obtain
\begin{align}
  &\lim_{N\to\infty} \Braket{f(\check{\bm{x}}_N)}_N^{\check{\eta}}
  = f(\bm{x}_\mathrm{max}^\eta), \label{eq:tr_f-rho_check}\\
  &\lim_{N\to\infty} \psi_N^{\check{\eta}}
  = \eta(\bm{x}_\mathrm{max}^\eta) - \s(\bm{x}_\mathrm{max}^\eta), \label{eq:psi_check}
\end{align}
where $\bm{x}_\mathrm{max}^\eta$ is the unique maximum point
of $\s(\bm{x})-\eta(\bm{\kappa};\bm{x})$.

\subsection{Properties of the SE for noncommutative set \texorpdfstring{$\hat{\bm{X}}_N$}{XN}} \label{sec:SE_noncommutative_hat}
Using the theorems presented in Section~\ref{sec:SE_noncommutative_hat-check},
we derive properties of the SE
for the noncommutative set 
$\hat{\bm{x}}_N$ from those of the SE for the commutative set 
$\check{\bm{x}}_N$
in Appendix~\ref{sec:properties_proof}.
Here, we present the results and their implications.

First, we find that, as in the commutative case, 
the expectation value of $\hat{\bm{x}}_N$ approaches $\bm{x}_\mathrm{max}^\eta$:
\begin{result} \label{property:x_N^eta}
\begin{align}
  x_{\indind{N}{i}}^{\hat{\eta}}(\bm{\kappa})
  \equiv \Tr \left[ \hat{x}_{\indind{N}{i}} \rho_N^{\hat{\eta}}(\bm{\kappa}) \right]
  \overset{\mathrm{TDL}}{\to} x_{\indind{\mathrm{max}}{i}}^\eta(\bm{\kappa}) \label{eq:x_N^eta}
\end{align}
\end{result}
Furthermore, we obtain the variance as
\begin{result} \label{property:var_x_N^eta}
\begin{align}
  \lim_{N\to\infty} \Tr \left[ \left( \hat{x}_{\indind{N}{i}} - x_{\indind{N}{i}}^{\hat{\eta}}(\bm{\kappa}) \right)^2 \rho_N^{\hat{\eta}}(\bm{\kappa}) \right] = 0 \label{eq:var_x_N^eta}
\end{align}
\end{result}
That is, 
all $\hat{x}_{\indind{N}{i}}$ have macroscopically definite values 
in the SE even when they are noncommutative.
This is in contrast to the {\CE},
in which some $\hat{x}_{\indind{N}{i}}$ has macroscopically large fluctuation
at the first-order phase transition point.

We can also show that the density matrix of the SE can be characterized
by the principle of equal probability
that was explained in Section~\ref{sec:SE_commutative}
because the argument is also valid for the noncommutative case.

From these results, it is concluded that
the density matrix $\hat{\rho}_N^\eta$ given by the SE
gives typical properties of the microstates
with $\bm{x} = \bm{x}_\mathrm{max}^\eta(\bm{\kappa}) + o(N^0)$.
Therefore, it gives the equilibrium state
specified by $\bm{x}_\mathrm{max}^\eta(\bm{\kappa})$.
Moreover, we will show in the next section that 
$\bm{x}_\mathrm{max}^\eta(\bm{\kappa})$ changes continuously 
as a function of $\bm{\kappa}$.
Hence, by choosing $\bm{\kappa}$
such that $\bm{x}_\mathrm{max}^\eta(\bm{\kappa})$ lies within a first-order phase transition region,
one can obtain a phase coexistence state.
Thus, by using the SE,
one can investigate microscopic structures of phase coexisting states, such as the phase interfaces.

Furthermore, 
the thermodynamic entropy density is obtained by the following formula:
\begin{result} \label{formula:entropy}
\begin{align}
  s_N^{\hat{\eta}}(\bm{\kappa})
  &\equiv \eta(\bm{\kappa};\bm{x}_N^{\hat{\eta}}(\bm{\kappa})) - \psi_N^{\hat{\eta}}(\bm{\kappa})
  \overset{\mathrm{TDL}}{\to} \s(\bm{x}_\mathrm{max}^\eta(\bm{\kappa})). \label{eq:entropy-formula}
\end{align}
\end{result}
Since $\eta$ is a known function,
one can obtain the thermodynamic entropy density
easily from the expectation value $\bm{x}_N^\eta$ and $\psi_N^\eta$. 
Finally, the intensive parameters are obtained by
\begin{result} \label{formula:tdforce}
\begin{align}
  \Pi_{\indind{N}{i}}^{\hat{\eta}}(\bm{\kappa})
  &\equiv \frac{\partial\eta}{\partial x_{\ind{i}}}(\bm{\kappa};\bm{x}_N^{\hat{\eta}})
  \overset{\mathrm{TDL}}{\to} \Pi_{\ind{i}}(\bm{x}_\mathrm{max}^\eta(\bm{\kappa})) \label{eq:tdforce-formula}
\end{align}
\end{result}
Since $\eta$ is a known function,
one can obtain the intensive parameters
just by calculating the expectation value $\bm{x}_N^\eta$,
without differentiating thermodynamic functions.

\subsection{Summary of properties when \texorpdfstring{$\hat{\bm{X}}_N$}{XN} is noncommutative}
\label{sec:summary_noncommutative}
To summarize this section,
all the properties derived in Section~\ref{sec:SE_commutative}
for the case where a set of additive observables $\hat{\bm{X}}_N$ 
is commutative
hold similarly for the case where $\hat{\bm{X}}_N$ is noncommutative.
In addition, 
the density matrix of the SE can be characterized by the principle of equal probability.
Therefore, by using the SE, we can correctly obtain the desired phase coexistence states of general quantum systems without any ad hoc procedures, such as devising boundary conditions.
We can also obtain genuine thermodynamic quantities,
such as the thermodynamic entropy density and the intensive parameter, easily from the SE.
Therefore, one can investigate microscopic structures and thermodynamic properties of phase coexistence states of general quantum systems.

\section{Parameter of squeezed ensemble} \label{sec:SE_parameter}
So far, we have investigated the properties
when the parameter $\bm{\kappa}$ is fixed to 
an arbitrary set of values. We now discuss the $\bm{\kappa}$
dependencies of the statistical-mechanical quantities given by the SE,
under conditions~\ref{cond:A}-\ref{cond:D}.
We will show that the equilibrium state described by the SE changes continuously
with respect to the parameter $\bm{\kappa}$ (Section~\ref{sec:parameter_rho})
and that the thermodynamic function obtained from the SE is
related to the thermodynamic entropy $\s$
via a generalization of the Legendre-Fenchel transformation (Section~\ref{sec:parameter_psi}).
We will also discuss connection with conventional generalized ensembles (Section~\ref{sec:genCE}).

\subsection{{\CE}} \label{sec:parameter_CE}
Suppose that the system does not undergo a phase transition,
so that its thermodynamic entropy density $\s$ is strongly concave.
In this case, we can take $\eta$ as a linear function,
$\eta(\bm{\kappa};\bm{x}) = \bm{\kappa}\cdot\bm{x}$,
which reduces to the {\CE}. Therefore, the {\CE} can be regarded
as a special form of the SE when a phase transition is absent.

To distinguish the {\CE} from the general SEs,
we write $c$ and $\bm{\pi} = (\pi_{\ind{0}}, \cdots, \pi_{\ind{m-1}})$
for $\eta$ and $\bm{\kappa}$ of the {\CE}, respectively.
That is, $\eta$ for the {\CE} is denoted as
\begin{align}
  \indexCE(\bm{\pi};\bm{x}) &= \bm{\pi}\cdot\bm{x}. \label{eq:CE_c}
\end{align}
Then it is seen from Eq.~\eqref{eq:tdforce-formula} that 
\begin{align}
  \pi_{\ind{i}} = \Pi_{\ind{i}}(\bm{x}_\mathrm{max}^\indexCE(\bm{\pi})). \label{eq:tdforce-formula_CE}
\end{align}
This equation simply states that the intensive parameters $\bm{\Pi}$ 
in the equilibrium state specified by $\bm{\pi}$ are just $\bm{\pi}$, 
as they should be.

In this particular case, 
$\bm{\kappa}$ coincides with the thermodynamic quantities $\bm{\Pi}$.
For general SEs, however, $\bm{\kappa}$ is not directly related to the thermodynamic quantities.
Nevertheless, one can calculate thermodynamic quantities from the SEs
via Eqs.~\eqref{eq:s-psi} and \eqref{eq:psi-s} as shown below,
as well as via Eqs.~\eqref{eq:entropy-formula} and \eqref{eq:tdforce-formula}.

\subsection{Parameter dependence of the density matrix of the SE} \label{sec:parameter_rho}
We consider the parameter dependence of the equilibrium state described by the SE.
As shown in Sections~\ref{sec:SE_commutative} and \ref{sec:SE_noncommutative},
the equilibrium state described by the SE
is specified by $\bm{x}_\mathrm{max}^\eta$ in the thermodynamic limit.
Hence, we investigate the $\bm{\kappa}$ dependence of $\bm{x}_\mathrm{max}^\eta$,
i.e., the function $\bm{x}_\mathrm{max}^\eta(\bm{\kappa})$.

Since $\bm{x}_\mathrm{max}^\eta$
is the maximum point of the strongly-concave function  
$\s(\bm{x})-\eta(\bm{\kappa};\bm{x})$, 
it is uniquely determined by
\begin{align}
  F_{\ind{i}}(\bm{\kappa};\bm{x}_\mathrm{max}^\eta) = 0
  \qquad (i=0, 1, \cdots, m-1),
\end{align}
where
\begin{align}
  F_{\ind{i}}(\bm{\kappa};\bm{x}) = \frac{\partial \s}{\partial x_{\ind{i}}}(\bm{x})
  - \frac{\partial \eta}{\partial x_{\ind{i}}}(\bm{\kappa};\bm{x})
  \quad (i=0, 1, \cdots, m-1).
\end{align}
From condition~\ref{cond:C}, $F_{\ind{i}}$ is a continuously differentiable function of $\bm{\kappa}$ and $\bm{x}$.
Hence, applying an implicit function theorem \cite{Rudin1976}, we find
\begin{result} \label{property:continuity}
  $\bm{x}_\mathrm{max}^\eta$ changes continuously with respect to $\bm{\kappa}$.
\end{result}
Therefore, the equilibrium state described by the SE changes continuously 
as a function of $\bm{\kappa}$.

This should be contrasted with the {\CE}
at a first-order phase transition point.
In that case, since condition~\ref{cond:B} is not satisfied,
the {\CE} is not an SE and $\bm{x}_\mathrm{max}^\indexCE$ is not unique.
Consequently, as passing through the first-order phase transition point,
$\bm{x}_\mathrm{max}^\indexCE$ switches discontinuously.
This means that the equilibrium state described by the {\CE}
changes discontinuously as a function of $\bm{\pi}$.
Therefore, the {\CE} is unable to describe the equilibrium states
in the first-order phase transition region.
Property~\ref{property:continuity} implies that such a discontinuous change does not occur in the SE.

\subsection{Thermodynamic function} \label{sec:parameter_psi}
Let us define the {\em thermodynamic function associated with $\eta$}
as the thermodynamic limit of $\psi_N^\eta(\bm{\kappa})$:
\begin{align}
  \psi^\eta(\bm{\kappa}) \equiv \lim_{N \to \infty} \psi_N^\eta(\bm{\kappa}). \label{eq:psi_eta}
\end{align}
Then we can rephrase Eq.~\eqref{eq:entropy-formula} using Eq.~\eqref{eq:x_N^eta} as
\begin{align}
  \psi^\eta(\bm{\kappa}) &= \inf_{\bm{x} \in \Omega} \left\{ \eta(\bm{\kappa};\bm{x})-\s(\bm{x}) \right\}. \label{eq:s-psi}
\end{align}
We can also invert this relation
owing to condition~\ref{cond:D},
as proven in Appendix~\ref{sec:genLeg_proof}:
\begin{result} \label{property:psi-s}
For all $\bm{x} \in \Omega$, we have
\begin{align}
  \s(\bm{x}) &= \inf_{\bm{\kappa} \in K} \left\{ \eta(\bm{\kappa};\bm{x})-\psi^\eta(\bm{\kappa}) \right\} \label{eq:psi-s}
\end{align}
\end{result}
Using this relation,
one can obtain the thermodynamic entropy $\s(\bm{x})$ directly from $\psi^\eta(\bm{\kappa})$ 
without knowing the function $\bm{x}^\eta(\bm{\kappa})$.
In this sense, $\psi^\eta(\bm{\kappa})$ is equivalent to the thermodynamic entropy.
Therefore, all thermodynamic functions are obtained from $\psi^\eta(\bm{\kappa})$.

From a physical point of view,
relation~\eqref{eq:s-psi}-\eqref{eq:psi-s} is a generalization of the equivalence of the entropy and the canonical thermodynamic function.
From a mathematical point of view, the relation is a generalization of the Legendre-Fenchel transformation.

\subsection{Connection with the conventional generalized ensembles} \label{sec:genCE}

Here we mention the connection with the conventional generalized ensembles.

Consider the SE for the particular choice of $\eta$,
\begin{align}
  \eta(\bm{\kappa};\bm{x}) &= \bm{\kappa}\cdot\bm{x} + g(\bm{x}), \label{eq:genCE}
\end{align}
where $g$ is a continuous function that is fixed independently of $\bm{\kappa}$.
This $\eta$ defines the so-called generalized canonical ensemble introduced by M. Costeniuc~\textit{et~al.} \cite{Costeniuc2005,Costeniuc2006}.
It includes generalized ensembles introduced in earlier studies
such as the Gaussian ensemble, which is obtained by taking $g$ as
\begin{align}
  g(\bm{x}) = \gamma \sum_i {x_{\ind{i}}}^2,
\end{align}
where $\gamma$ is a fixed positive constant
independent of $\bm{\kappa}$.
In fact, this choice of $g$ gives the density matrix that can be written as
\begin{align}
  \hat{\rho}_N^\eta(\bm{\kappa})
  \propto \exp \left[- N \gamma \sum_i \left( \hat{x}_{\indind{N}{i}} + \frac{\kappa_{\ind{i}}}{2\gamma} \right)^2 \right],
\end{align}
which has the same form as the Gaussian ensemble
originally proposed by Hetherington \cite{Hetherington1987,Challa1988_PRL,Challa1988_PRA,Johal2003,Costeniuc2005}.

In the generalized canonical ensemble, Eq.~\eqref{eq:genCE}, 
the parameter $\kappa_{\ind{i}}$ couples with $x_{\ind{i}}$ linearly
in order for $\s(u)-g(u)$ to be
the Legendre-Fenchel transform of the thermodynamic function given by the ensemble.
Therefore, in quantum systems,
it is necessary to compute the matrix exponential
each time the parameters are changed.
In numerical calculations, the matrix exponential is usually evaluated
by utilizing the Suzuki-Trotter decomposition, but this method has drawbacks such as the Trotter error.

In comparison, 
the parameter dependence is more general in the SE.
Even when $\kappa_{\ind{i}}$ couples with $x_{\ind{i}}$ {\em nonlinearly}
the thermodynamic function given by the SE is equivalent to the thermodynamic entropy as shown in Section~\ref{sec:parameter_psi}.
This generality allows us to choose $\eta$
that is greatly convenient for practical use.

To illustrate this point let us take a simple example of
the case where the equilibrium state is specified
only by the internal energy (i.e., $m=1$).
We can take $\eta=-2\kappa_{\ind{0}}\log(l-x_{\ind{0}})$ and $K=(0,\infty)$,
where $l$ is an arbitrary constant 
such that $Nl \geq$ the largest eigenvalue of 
$\hat{H}_N$ \cite{Sugiura2012,Yoneta2019}.
This choice of $\eta$ is particularly convenient for quantum systems,
because a series of density matrices given by the SE
\begin{align}
  \hat{\rho}_N^\eta(\kappa_{\ind{0}}=k/N) \propto ( Nl - \hat{H}_N )^{2k}
  \qquad \left( k=0,1,2,\cdots \right)
\end{align}
can be constructed sequentially
by simply multiplying $(Nl-\hat{H}_N)$ repeatedly:
\begin{align}
  \hat{\rho}_N^\eta((k+1)/N) &\propto ( Nl - \hat{H}_N )^2 \times \hat{\rho}_N^\eta(k/N).
\end{align}
Therefore, when using this SE, the calculation of the matrix exponential,
which is unavoidable when using the generalized canonical ensemble,
is no longer necessary.
Although the SEs that can be constructed in this way
are limited to those at discrete points $\kappa_{\ind{0}} = k/N (k=0,1,2,\cdots)$,
these states are sufficiently dense for large $N$
in the thermodynamic state space $\Omega$.
Therefore, one can obtain all the thermodynamic properties of the system.
In fact, using Eq.~\eqref{eq:tdforce-formula}, we have
\footnote{We write $f(N)=O(g(N))$ when $|f|$ is bounded above by $g$ up to a constant factor $K$ in the thermodynamic limit such that $\displaystyle |f(N)| \leq K g(N)$ as $N \to \infty$.}
\begin{align}
  \frac{\partial x_{\indind{N}{0}}^\eta}{\partial \kappa_{\ind{0}}}
  = \frac{2}{\displaystyle (l-x_{\indind{N}{0}}^\eta)\frac{\partial \beta}{\partial x_{\ind{0}}}-\beta} + o(N^0)
  = O(N^0).
\end{align}
Therefore,
\begin{align}
  x_{\indind{N}{0}}^\eta((k+1)/N) - x_{\indind{N}{0}}^\eta(k/N) = O(N^{-1}),
\end{align}
which implies $x_{\indind{N}{0}}$ is dense in the thermodynamic limit.

\subsection{Summary of Parameter Dependence} \label{sec:Summary_parameter}
To summarize this section,
although the parameter $\bm{\kappa}$ of the SE
is not directly related to familiar thermodynamic quantities,
such as the intensive parameter,
it corresponds to a point in the thermodynamic state space,
and, by varying it, we can investigate the properties of a series of equilibrium states.
In particular, by simply computing the thermodynamic function given by the SE
as a function of $\bm{\kappa}$,
one can obtain all the thermodynamic properties of the system.
As compared with the conventional generalized ensembles,
it is easier to compute statistical-mechanical quantities while varying 
$\bm{\kappa}$ because the SE allows a much wider choice of the function $\eta$.

\section{Usage of the SE} \label{sec:SE_usage}
Let us explain how to use the SE when applying it to concrete calculations.

First, choose $\eta$ which satisfies conditions~\ref{cond:A}-\ref{cond:D},
depending on the physical situations and the numerical method to be used.
To obtain accurate results,
special attention should be paid to condition~\ref{cond:B}.
When one can deal with the thermodynamic limit,
it is sufficient
if $\s(\bm{x})-\eta(\bm{\kappa};\bm{x})$ is strongly concave
as required by condition~\ref{cond:B}.
However, in practical calculations for finite $N$ that is not so large,
it is necessary to take the convexity of $\eta$ strong enough
such that it overcomes the concavity breaking
caused by finite-size effects \cite{Yoneta2019}.

Next, construct the SE associated with $\eta$ chosen above
according to Eq.~\eqref{eq:rho_N^eta} and \eqref{eq:psi_N^eta}.
Then, the density matrix $\hat{\rho}_N^\eta$
always gives a macroscopically definite equilibrium state.
Therefore, by calculating the expectation values in the density matrix,
one can obtain the equilibrium values of the local observables and additive observables
for any equilibrium state in the thermodynamic state space,
including the first-order phase transition regions.
 
One can also calculate genuine thermodynamic quantities of the equilibrium state.
In particular, 
using Eq.~\eqref{eq:tdforce-formula},
one can obtain the intensive parameters
just by calculating the expectation value of $\hat{\bm{x}}_N$,
without resorting to numerical differentiation.

Furthermore, one can obtain other genuine thermodynamic quantities,
such as the thermodynamic entropy density $\s$,
by using Eqs.~\eqref{eq:entropy-formula} or \eqref{eq:psi-s}.
In particular, by using Eq.~\eqref{eq:psi-s},
one can obtain all the thermodynamic properties of the system
by simply computing $\psi_N^\eta$ as a function of $\bm{\kappa}$.

In the above manner,
one can fully investigate microscopic structures and thermodynamic properties
of general quantum systems, 
even for phase coexistence states without any ad hoc procedures
such as finding and imposing clever boundary conditions according to the phases to be coexisted.
We will demonstrate this point in Section~\ref{sec:d2Ising}.

\section{Free Spins} \label{sec:freeSpin}
In this and next sections, we apply our formulation 
to concrete examples
and demonstrate its validity and availability.
As the first example, we apply our formulation to the free spins,
which was mentioned in Section~\ref{sec:noncommutativity},
in this section.
Although this system does not undergo a first-order phase transition,
it will help to understand properties of the SE with noncommutative $\hat{\bm{x}}_N$.

\subsection{Comparison with the canonical ensemble} \label{sec:comp_CE_free_spins}
As discussed in Section~\ref{sec:CE}, 
the {\CE} always gives the correct thermodynamic functions.
Furthermore, in this simple model
the {\CE} also gives the correct equilibrium state
because the phase transition does not occur.
By contrast, as discussed in Section~\ref{sec:noncommutativity},
we cannot employ either the {\MCE} or the RE
to the free spin model defined by Eqs.~\eqref{eq:X0_free_spin} and \eqref{eq:X1_free_spin}
because $\hat{X}_{\indind{N}{0}}$ and $\hat{X}_{\indind{N}{1}}$
do not commute.
Let us confirm that the SE gives the correct equilibrium sate and thermodynamic functions even in this case
by comparing them with those obtained from the {\CE}.

We have a wide choice of the parameter space $K$ of $\bm{\kappa}$ and the functional form of $\eta$.
We here choose them as
\begin{align}
  K &= (0,\infty)\times\mathbb{R},\\
  \eta(\bm{\kappa};\bm{x}) &= \kappa_{\ind{0}} x_{\ind{0}} + \frac{1}{2} \lambda (x_{\ind{1}}-\kappa_{\ind{1}})^2, \label{eq:freeSpin_eta}
\end{align}
where $\lambda$ is a positive constant.
Since this $\eta$ satisfies all conditions~\ref{cond:A}-\ref{cond:D},
this is an SE,
and therefore $\hat{\bm{x}}_N$ has a macroscopically definite value in its density matrix,
as shown in Eq.~\eqref{eq:var_x_N^eta}.

On the other hand,
the {\CE} for this system is obtained
by taking $\eta$ as Eq.~\eqref{eq:CE_c}, i.e.,
\begin{align}
  \indexCE(\bm{\pi};\bm{x}) &= \pi_{\ind{0}} x_{\ind{0}} + \pi_{\ind{1}} x_{\ind{1}}.
\end{align}
This $\eta$ for the {\CE} also satisfies conditions~\ref{cond:A}-\ref{cond:D}
in this simple system because a first-order phase transition is absent.
Hence, $\hat{\bm{x}}_N$ has a macroscopically definite value
also in the density matrix of the {\CE}.
Therefore, if one takes the values of $\bm{\kappa}$ and $\bm{\pi}$ 
appropriately (as Eq.~\eqref{eq:Pi^c=Pi^eta} below), 
the {\CE} and the SE represent the same equilibrium state.
That is, 
the SE associated with $\eta$ in Eq.~\eqref{eq:freeSpin_eta}
is equivalent to the {\CE}
not only for the thermodynamic function
but also for the density matrix,
in the thermodynamic limit.

We demonstrate these facts numerically.
In numerical calculations, 
analytic results are used for the {\CE}.
For the SE, we use 
the thermal pure quantum formulation \cite{Sugiura2012,Sugiura2013}, 
in which 
$\hat{\rho}_N^\eta(\bm{\kappa})$
is replaced with the (unnormalized) pure state,
\begin{align}
  e^{-N\eta(\bm{\kappa};\hat{\bm{x}}_N)/2} \ket{\psi_\mathrm{rand}}, \label{eq:rho_N^eta_TPQ}
\end{align}
which gives the same results as $\hat{\rho}_N^\eta(\bm{\kappa})$
and $  \psi_N^\eta (\bm{\kappa})$
with exponentially small probability of errors.
Here, $\ket{\psi_\mathrm{rand}}$ is a Haar-random vector in the whole Hilbert space.
Note that this replacement is valid not only for the above $\eta$ but also 
for all SEs.

\subsection{Density Matrix}
First, we demonstrate that the density matrix $\hat{\rho}_N^\eta$ given by the SE
describes the equilibrium state
which has the intensive parameters given by Eq.~\eqref{eq:tdforce-formula}.

Since this model does not undergo the first-order phase transition,
the equilibrium state is uniquely specified
by the set of intensive parameters $\bm{\Pi}$.
Therefore, in order for the {\CE} and the SE
to describe the same equilibrium state
in the thermodynamic limit,
we choose the parameters $\bm{\pi}$ and $\bm{\kappa}$ so that
\begin{align}
  \lim_{N\to\infty} \bm{\Pi}_N^\indexCE(\bm{\pi}) = \lim_{N\to\infty} \bm{\Pi}_N^\eta(\bm{\kappa}).
\label{eq:Pi^c=Pi^eta}
\end{align}
In concrete calculations, we fix $\bm{\kappa}$
and calculate statistical-mechanical quantities in the SE,
including $\bm{\Pi}_N^\eta(\bm{\kappa})$.
Then we set $\bm{\pi}$ in such a way
that $\bm{\pi}=\bm{\Pi}_N^\eta(\bm{\kappa})$ for each $N$.
By doing so, we have
\begin{align}
  \bm{\Pi}_N^\indexCE(\bm{\pi}=\bm{\Pi}_N^\eta(\bm{\kappa})) = \bm{\Pi}_N^\eta(\bm{\kappa})
\end{align}
from Eq.~\eqref{eq:tdforce-formula_CE}.
Therefore,
the {\CE} and the SE should describe the same equilibrium state in the thermodynamic limit.

To demonstrate the equivalence for the density matrix,
we have plotted in Fig.~\ref{fig:freeSpin_N-z}
the $N$ dependence of the difference in the expectation value of 
$\hat{x}_{\indind{N}{1}}=\frac{1}{N}\sum_{i=1}^N \hat{\sigma}_i^y$
between the density matrices given by the {\CE} and the SE,
\begin{align}
  \left| x_{\indind{N}{1}}^\eta(\bm{\kappa})
       - x_{\indind{N}{1}}^\indexCE(\bm{\Pi}_N^\eta(\bm{\kappa})) \right|.
\end{align}
We find that the difference decreases with increasing $N$, 
proportionally to $N^{-1}$.
We thus confirm that the SE gives the same expectation value in the thermodynamic limit.

\begin{figure}
  \centering
  \includegraphics[keepaspectratio, width=\linewidth]{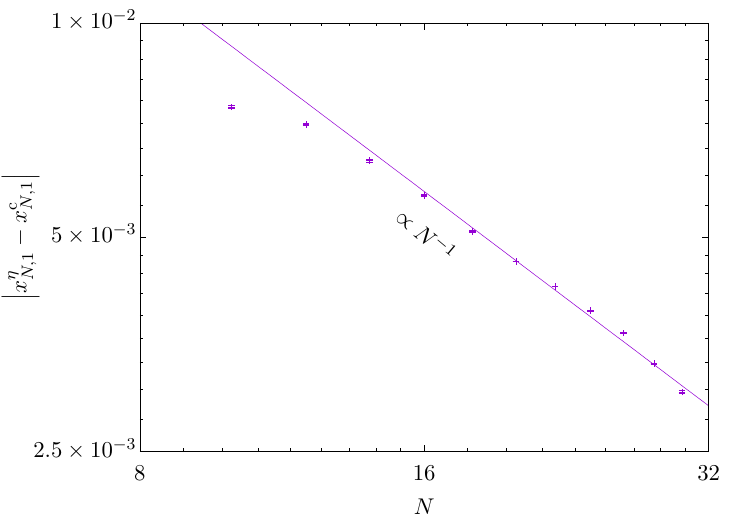}
  \caption{$N$ dependence of the difference in the expectation value of 
  $\hat{x}_{\indind{N}{1}}=\frac{1}{N}\sum_{i=1}^N \hat{\sigma}_i^y$}
  between the density matrices given by the {\CE} and the SE associated with $\eta$ in Eq.~\eqref{eq:freeSpin_eta},
  for the free spin model defined by Eqs.~\eqref{eq:X0_free_spin} and \eqref{eq:X1_free_spin}.
  \label{fig:freeSpin_N-z}
\end{figure}

In addition, we have plotted in Fig.~\ref{fig:freeSpin_N-var}
the $N$ dependence of the variances of $\hat{x}_{\indind{N}{0}}$ and $\hat{x}_{\indind{N}{1}}$ in the SE.
We find that the variances decay as $O(N^{-1})$.
This is consistent with the results obtained by Laplace's approximation (see Section~\ref{sec:densityMatrix}).
Therefore, it is confirmed that
the SE gives the macroscopically definite state
even for the noncommutative case.
This is in sharp contrast with the RE,
for which the variance of $\hat{x}_{\indind{N}{0}}$ remains finite
even in the thermodynamic limit
due to the noncommutativity of $\hat{x}_{\indind{N}{0}}$ and $\hat{x}_{\indind{N}{1}}$,
as discussed in Section~\ref{sec:noncommutativity}.

\begin{figure}
  \centering
  \includegraphics[keepaspectratio, width=\linewidth]{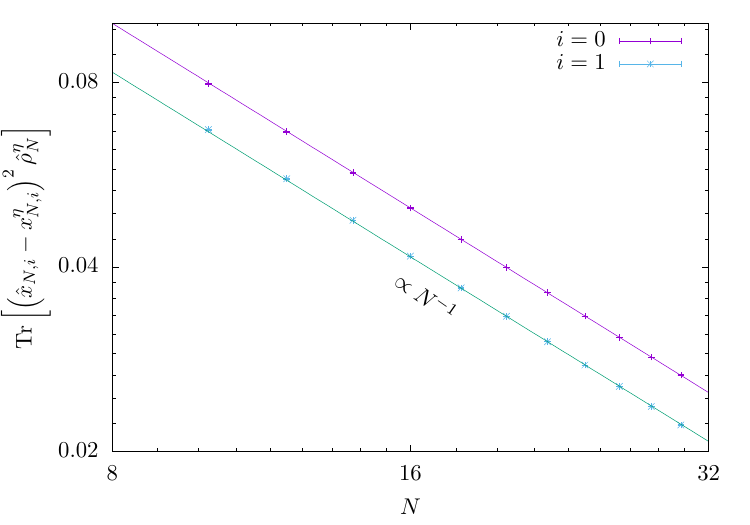}
  \caption{$N$ dependence of the variances of
  $\hat{x}_{\indind{N}{0}}=\frac{1}{N}\sum_{i=1}^N \hat{\sigma}_i^x$
  and $\hat{x}_{\indind{N}{1}}=\frac{1}{N}\sum_{i=1}^N \hat{\sigma}_i^y$
  in the SE associated with $\eta$ in Eq.~\eqref{eq:freeSpin_eta}
  with $\lambda=0.5$ and $\bm{\kappa}=(1,0.5)$,
  for the free spin model defined by Eqs.~\eqref{eq:X0_free_spin} and \eqref{eq:X1_free_spin}.}
  \label{fig:freeSpin_N-var}
\end{figure}

These results demonstrate that the SE successfully gives the correct equilibrium state
even for systems with noncommutative additive observables.

\subsection{Thermodynamic function}
Next, we demonstrate that the thermodynamic function $\psi^\eta(\bm{\kappa})$ of the SE
is equivalent to the thermodynamic entropy via Eqs.~\eqref{eq:s-psi}-\eqref{eq:psi-s}.

The free spin model defined by Eqs.~\eqref{eq:X0_free_spin} and \eqref{eq:X1_free_spin}
has the normal property that
the canonical thermodynamic function $\psi^\indexCE(\bm{\pi})$
is equivalent to the thermodynamic entropy:
\begin{align}
  \s(\bm{x}) &= \inf_{\bm{\pi} \in \mathbb{R}^2} \left\{ \bm{\pi}\cdot\bm{x} - \psi^\indexCE(\bm{\pi}) \right\}. \label{eq:freeSpin_s-psi_CE}
\end{align}
According to Eq.~\eqref{eq:s-psi},
this $\s$ is equivalent to the thermodynamic function $\psi^\eta(\bm{\kappa})$ of the SE:
\begin{align}
  \psi^\eta(\bm{\kappa}) = \inf_{\bm{x} \in \Omega} \left\{ \eta(\bm{\kappa};\bm{x}) - \s(\bm{x}) \right\}.
\end{align}
To demonstrate this equivalence, we have plotted in Fig.~\ref{fig:freeSpin_N-psi} the $N$ dependence of
\begin{align}
  \left| \psi_N^\eta(\bm{\kappa}) - \inf_{\bm{x} \in \Omega} \left\{ \eta(\bm{\kappa};\bm{x}) - \s(\bm{x}) \right\} \right|,
\end{align}
where $\s(\bm{x})$ is calculated from
the canonical thermodynamic function using Eq.~\eqref{eq:freeSpin_s-psi_CE}.
We find that the difference decreases with increasing $N$,
nearly proportionally to $N^{-1}$.
% in almost proportion to $N^{-1}$.
Therefore, 
it is confirmed that
$\psi^\eta(\bm{\kappa})$ is equivalent to the thermodynamic entropy density
in the thermodynamic limit.

\begin{figure}
  \centering
  \includegraphics[keepaspectratio, width=\linewidth]{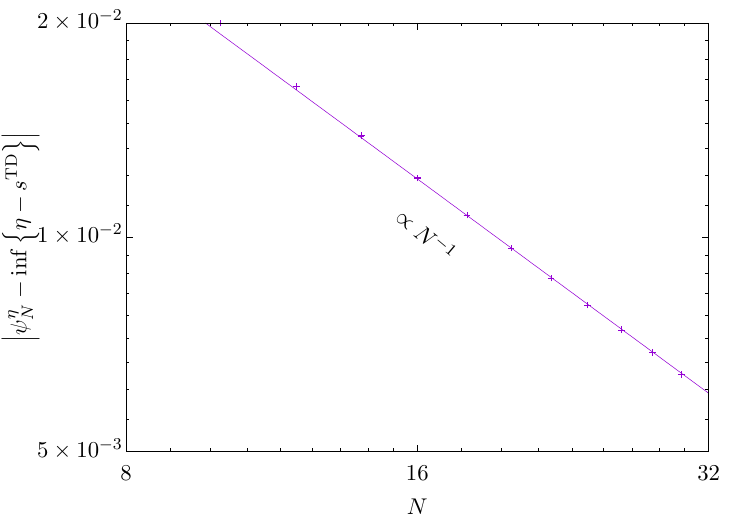}
  \caption{$N$ dependence of the difference
  between the thermodynamic function $\psi_N^\eta$ given by the SE associated with $\eta$ in Eq.~\eqref{eq:freeSpin_eta}
  with $\lambda=0.5$ and $\bm{\kappa}=(1,0.5)$
  and the generalized Legendre-Fenchel transformation of the thermodynamic entropy $\s$,
  for the free spin model defined by Eqs.~\eqref{eq:X0_free_spin} and \eqref{eq:X1_free_spin}.}
  \label{fig:freeSpin_N-psi}
\end{figure}

\section{Application to the two-dimensional transverse field Ising model} \label{sec:d2Ising}
We apply our formulation to a system
which undergoes a first-order phase transition with an order parameter
that does not commute with the Hamiltonian
and thereby demonstrate numerically that the SE indeed gives 
the density matrix corresponding to
a phase coexistence state
even in such a case which cannot be treated by conventional ensembles.

To be concrete, we apply our formulation
to the two-dimensional transverse field Ising model
on the square lattice,
defined by the Hamiltonian
\begin{align}
  \hat{H}_N = \hat{X}_{\indind{N}{0}}
  \equiv - J \sum_{\langle  i, j \rangle} \hat{\sigma}_i^z \hat{\sigma}_j^z
  - g \sum_{i} \hat{\sigma}_i^x.
\end{align}
Here, $\langle i, j \rangle$ denotes
the nearest neighbor sites $i$ and $j$.
We take the coupling $J$ positive
and measure energy in units of $J$.

\subsection{Inapplicability of conventional ensembles}
This model has an ordered phase
for small $g$ and large $\beta$ $(=\Pi_{\ind{0}})$
\cite{Elliott1971,Pfeuty1971,Nagai1987,Rieger1999,Blote2002,Nakamura2003,Jongh1998}.
Its order parameter is given by
\begin{align}
  \hat{X}_{\indind{N}{1}} &= \sum_i \hat{\sigma}_i^z,
\end{align}
and its conjugate field is the $z$-component of a magnetic field,
$f_{\ind{1}}=-\Pi_{\ind{1}}/\Pi_{\ind{0}}$
(which should not be confused with the $x$-component, $g$).
As $f_{\ind{1}}$ is changed,
while temperature is fixed at a value lower than the critical temperature, 
a first-order phase transition occurs at $f_{\ind{1}}=0$.
There, $x_{\ind{1}}$ changes discontinuously
as a function of $f_{\ind{1}}$.
(This is called the magnetic-field-driven transition 
as opposed to the temperature-driven transition.)
This indicates that phase coexistence states are realized at $f_{\ind{1}}=0$,
and the proportions of different phases cannot be specified by $f_{\ind{1}}$.
Consequently, the {\CE}
\begin{align}
  \hat{\rho}_N^\indexCE \propto \exp [- \beta (\hat{H}_N - f_{\ind{1}} \hat{X}_{\indind{N}{1}})]
\end{align}
cannot give the phase coexistence states
(unless some trick such as an artificial boundary condition is used), 
as discussed in Section~\ref{sec:introduction}.

In order to obtain phase coexistence states
it is necessary to employ an ensemble
with macroscopically well-defined order parameter.
However, neither the {\MCE} nor the RE is applicable
because $\hat{H}_N$ and $\hat{X}_{\indind{N}{1}}$ do not commute with each other when $g \neq 0$,
as pointed out in Section~\ref{sec:noncommutativity}.

We will show in Section~\ref{sec:d2qIsing}
that the SE solves this fundamental problem
and successfully describes the first-order phase transition.
We also show in Section~\ref{sec:d2cIsing} that
the SE has advantages even when this fundamental problem is absent (i.e., when $g=0$).

\subsection{Zero transverse field} \label{sec:d2cIsing}
We first consider the case of $g = 0$,
in which the model reduces to the classical Ising model
and hence the fundamental problem mentioned above is absent.
We will show that the SE is advantages over the {\CE} even in this case.

\subsubsection{Methods}
We choose the parameter space $K$ of $\bm{\kappa}$ and $\eta$ as
\begin{align}
  K &= (0,\infty)\times\mathbb{R},\\
  \eta(\bm{\kappa};\bm{x}) &= \kappa_{\ind{0}} x_{\ind{0}} + \frac{1}{2} \lambda (x_{\ind{1}}-\kappa_{\ind{1}})^2, \label{eq:d2cIsing_eta}
\end{align}
where $\lambda$ is a positive constant.
As will be described below, this choice of $\eta$ is convenient
for classical systems which undergo first-order phase transitions.

In order for the order parameter $x_{\ind{1}}$ to have a definite value,
the convexity of $\eta$ as a function of $x_{\ind{1}}$ must be strong enough.
This condition is particularly important for finite systems.
In the analysis in Section~\ref{sec:SE_commutative},
we have used the fact that $\sigma_N-\eta$
can be approximated by the strongly concave function $\s-\eta$ in the thermodynamic limit.
However, in finite systems 
the concavity of the microcanonical entropy $\sigma_N$ is broken
near the first-order phase transition region due to finite-size effects
\cite{Schulman1980,Binder1981}.
Therefore, if the convexity of $\eta$ is weak,
$\sigma_N-\eta$ can be a nonconcave function for finite $N$,
and the analysis using Laplace's approximation breaks down.
To overcome such finite-size effects, we take the convexity of $\eta$ strong enough.
Since the second order derivative of $\eta$ with respect to $x_{\ind{1}}$ is $\lambda$,
the strength of the convexity of $\eta$ can be freely adjusted by tuning the value of $\lambda$. 
That is, for sufficiently large $\lambda$,
$\sigma_N-\eta$ becomes strongly concave,
and the order parameter has a definite value.
In fact, Eq.~\eqref{eq:p_N^eta_Gaussian} gives
\begin{align}
  &\Tr [(\hat{x}_{\indind{N}{1}}-x_{\indind{N}{1}}^\eta)^2\hat{\rho}_N^\eta] \nonumber\\
  &= N^{-1} \left[
    \lambda-\frac{\partial^2 \s}{{\partial x_{\ind{1}}}^2}(\bm{x}_\mathrm{max}^\eta)
    + \frac{
      \left(\frac{\partial^2 \s}{\partial x_{\ind{0}} \partial x_{\ind{1}}}(\bm{x}_\mathrm{max}^\eta)\right)^2
    }{
      \frac{\partial^2 \s}{{\partial x_{\ind{0}}}^2}(\bm{x}_\mathrm{max}^\eta)
    }
  \right]^{-1} + o(N^{-1}). \label{eq:d2cIsing_var}
\end{align}
Therefore, the variance of $\hat{x}_{\indind{N}{1}}$ scales as $O(N^{-1}\lambda^{-1})$ for sufficiently large $N$ and $\lambda$.

From Eq.~\eqref{eq:tdforce-formula},
the inverse temperature $\beta = \Pi_{\ind{0}}$
is given directly by $\kappa_{\ind{0}}$ as
\begin{align}
  \beta_N^\eta \equiv \Pi_{\indind{N}{0}}^\eta = \kappa_{\ind{0}}. \label{eq:d2cIsing_Pi0}
\end{align}
Equation~\eqref{eq:tdforce-formula} also gives
the intensive parameter conjugate to $X_{\ind{1}}$ as
\begin{align}
  \Pi_{\indind{N}{1}}^\eta = \lambda (x_{\indind{N}{1}}^\eta-\kappa_{\ind{1}}), \label{eq:d2cIsing_Pi1}
\end{align}
which can easily be calculated from 
the expectation value $x_{\indind{N}{1}}^\eta$.
By combining this with Eq.~\eqref{eq:d2cIsing_Pi0},
we can calculate 
the magnetic field conjugate to the order parameter as
\begin{align}
  f_{\indind{N}{1}}^\eta \equiv - \frac{\lambda (x_{\indind{N}{1}}^\eta-\kappa_{\ind{1}})}{\kappa_{\ind{0}}}. \label{eq:d2cIsing_f1}
\end{align}
Using these formulas, one obtains the intensive parameters
without differentiating the thermodynamic functions numerically.

We calculate the statistical-mechanical quantities in the SE
using the replica exchange Monte Carlo calculations \cite{Hukushima1996},
in which the acceptance probability can be easily computed
using the Metropolis algorithm \cite{Metropolis1953}
in almost the same manner as in the {\CE}.
In the SE, 
the equilibrium state changes continuously in $\bm{\kappa}$
even in the phase transition region,
as explained in Section~\ref{sec:SE_parameter}.
Therefore, 
the exchange of configurations between adjacent replicas is accepted with high probability,
and the replica exchange method works well
(as in the case of the ensembles
with macroscopically well-defined internal energy
in the temperature-driven first-order phase transition region
\cite{Kim2010,Schierz2016}.)
By contrast, in the {\CE}, different phases are separated
by the free energy barrier \cite{Schierz2016}
and the equilibrium state changes discontinuously in the parameters of the ensemble.
This greatly degrades the performance of the replica exchange method in the {\CE}.

In the SE, 
boundary conditions can be imposed arbitrarily
because the order parameter has a definite value without
imposing artificial boundary conditions.
In the present numerical simulations,
we employ periodic boundary conditions.
Owing to this choice,
we can eliminate surface effects,
which are of the same order as the effects of the phase interfaces.
That is, the SE enables us to study the phase interfaces
without suffering from the surface effects.

\subsubsection{Results}
We have plotted in Fig.~\ref{fig:d2cIsing_L-var}
the $N$ dependence of the variances of $\hat{x}_{\indind{N}{1}}$ in the {\CE} and in the SE.
In order to compare them in the first-order phase transition region
(where $\beta > \frac{1}{2} \log (1+\sqrt{2}) \simeq 0.44$ and $f_{\ind{1}} = 0$),
we focus on the equilibrium state with $\beta=0.45$ and $f_{\ind{1}}=0$.
Accordingly, we take $\bm{\kappa}=(0.45,0)$ so that $\beta_N^\eta=0.45$ and $f_{\indind{N}{1}}^\eta=0$
from Eqs.~\eqref{eq:d2cIsing_Pi0} and \eqref{eq:d2cIsing_f1}.

It is seen that in the case of the {\CE} the variance remains finite even in the thermodynamic limit.
This implies that the {\CE} gives a statistical mixture of macroscopically distinct states
which have the same values of the intensive parameters.
In fact, as shown in the upper insets of Fig.~\ref{fig:d2cIsing_L-var},
there are two typical spin configurations for the {\CE},
which are macroscopically homogeneous states
(i.e., single phases) with positive or negative $x_{\ind{1}}$.

By contrast, in the case of the SE, the variance vanishes in the thermodynamic limit as $O(N^{-1})$.
This implies that the SE successfully gives a phase coexistence state,
as can be seen from the lower insets of Fig.~\ref{fig:d2cIsing_L-var}.
We have thus confirmed that the SE is a microcanonical-like ensemble
and gives the phase coexistence state.
\begin{figure}
  \centering
  \includegraphics[keepaspectratio, width=\linewidth]{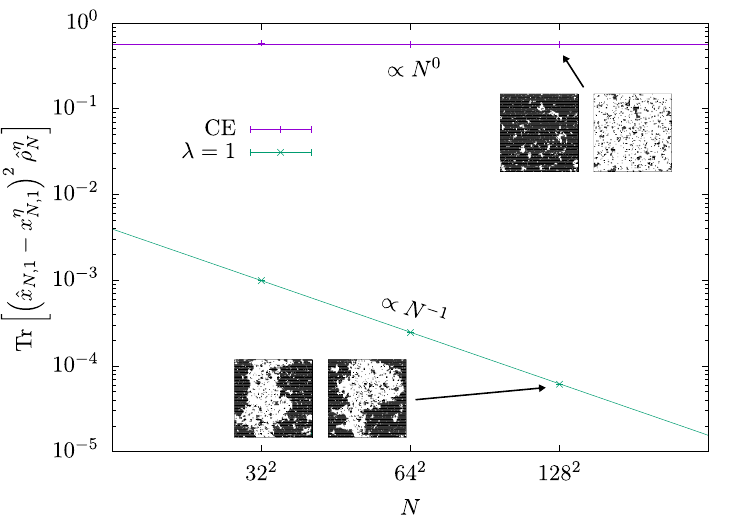}
  \caption{$N$ dependence of the variances of $\hat{x}_{\indind{N}{1}}$
  in the {\CE} and in the SE associated with $\eta$ in Eq.~\eqref{eq:d2cIsing_eta}
  at a first-order phase transition point ($\beta=0.45$ and $f_{\ind{1}}=0$)
  of the two-dimensional classical Ising model.
  Insets show typical snapshots of the Monte Carlo simulations
  (white represents $\braket{\hat{\sigma}_i^z}=+1$
  and black represents $\braket{\hat{\sigma}_i^z}=-1$).
  Two snapshots are shown for each ensemble.}
  \label{fig:d2cIsing_L-var}
\end{figure}

In order to illustrate how the convexity of $\eta$ controls the variance of $\hat{x}_{\indind{N}{1}}$,
we have plotted the variance as a function of $\lambda$ in Fig.~\ref{fig:d2cIsing_lambda-var}.
It is seen that the variance becomes smaller as the convexity of $\eta$ becomes stronger.
This is consistent with the analytical result of Eq.~\eqref{eq:d2cIsing_var}.

\begin{figure}
  \centering
  \includegraphics[keepaspectratio, width=\linewidth]{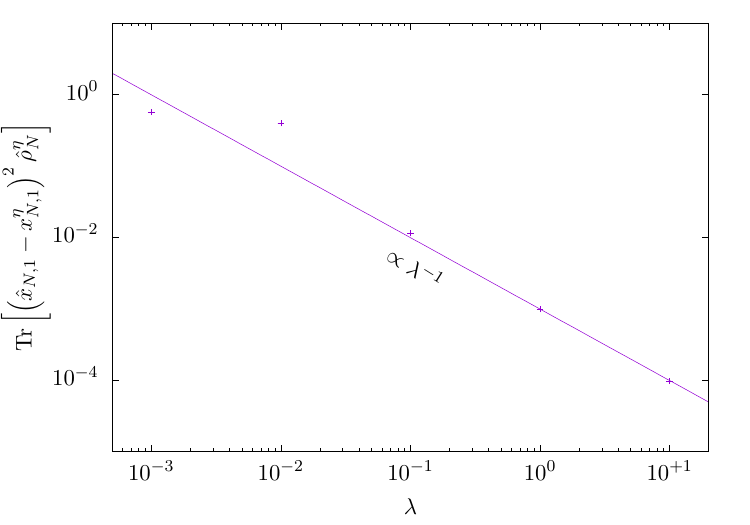}
  \caption{Relation between the strength $\lambda$ of the convexity of $\eta$
  defined by Eq.~\eqref{eq:d2cIsing_eta}
  and the variance of $\hat{x}_{\indind{N}{1}}$ in the SE
  at the first-order phase transition point ($\beta=0.45$ and $f_{\ind{1}}=0$)
  of the two-dimensional classical Ising model with $N=32^2$ spins.}
  \label{fig:d2cIsing_lambda-var}
\end{figure}

To see more details, we have plotted in Fig.~\ref{fig:d2cIsing_ordered_tdf-op}
the relation between the magnetic field $f_{\indind{N}{1}}^\eta$
and the order parameter $x_{\indind{N}{1}}^\eta$,
in the ordered phase for the {\CE} and 
for the SEs with various values of $\lambda$.
Again, we take $\kappa_{\ind{0}}=0.45$.

When the {\CE} is used,
$x_{\indind{N}{1}}^\indexCE=\Tr[\hat{x}_{\indind{N}{1}}\hat{\rho}_N^\indexCE]$ 
cannot be a multivalued function of $f_{\ind{1}}$
(because the {\CE} is a function of $f_{\ind{1}}$)
and hence varies monotonically and continuously.
The singularities of thermodynamic functions
in the first-order phase transition region
are almost smeared out
due to the large fluctuation inherent in the {\CE} in such a region.
This makes it difficult to identify the phase transition and to determine its order \cite{Stump1987,Huller1992,Huller1994}.

This should be contrasted with the results of the SEs,
which show that $x_{\indind{N}{1}}^\eta$ is a multivalued,
S-shaped function of $f_{\indind{N}{1}}^\eta$
for sufficiently large $\lambda$.
This is a manifestation of the concavity breaking of the microcanonical entropy
due to the finite-size effects of the phase interface \cite{Challa1988_PRA,Yoneta2019}.
It is thus confirmed that
thermodynamic anomalies are correctly obtained
and the phase transition can be detected directly by the SE.

It is also seen from Fig.~\ref{fig:d2cIsing_ordered_tdf-op}
that the results of the SE become insensitive to the magnitude of $\lambda$
for sufficiently large $\lambda$ ($\gtrsim 1$).
Therefore, one can use an arbitrary value of $\lambda$
as long as it is sufficiently large.
This is understood from the derivation in Section~\ref{sec:SE_commutative}:
As long as $\sigma_N-\eta$ is well approximated by the strongly concave function $\s-\eta$,
the probability distribution $p_N^\eta$ has a sharp peak,
and the statistical-mechanical quantities do not depend on the details of $\eta$.

\begin{figure}
  \centering
  \includegraphics[keepaspectratio, width=\linewidth]{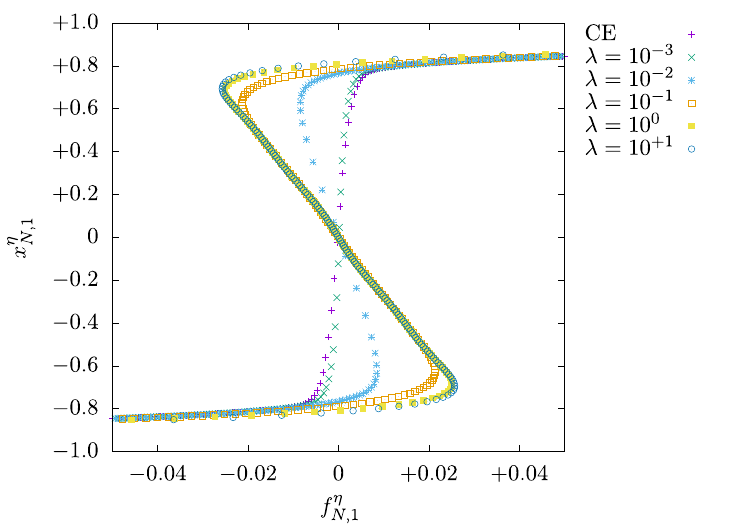}
  \caption{Relation between the magnetic field and the order parameter
  in the vicinity of the first-order phase transition point ($\beta=0.45$ and $f_{\ind{1}}=0$)
  of the two-dimensional classical Ising model with $N=32^2$ spins
  for the {\CE} and the SEs associated with $\eta$ in Eq.~\eqref{eq:d2cIsing_eta}.}
  \label{fig:d2cIsing_ordered_tdf-op}
\end{figure}

We have thus confirmed that by
using the SE
one can obtain phase coexistence states
and finite-size effects due to phase interfaces
for the commutative case.
These results also demonstrate that Laplace's approximation used in Section~\ref{sec:SE_commutative} 
will become exact in the thermodynamic limit.

\subsection{Finite transverse field} \label{sec:d2qIsing}
To demonstrate that the SE gives phase coexistence states 
with an order parameter noncommutative with the Hamiltonian,
we consider the case of $g \neq 0$ and take a moderate value, $g=1$.

\subsubsection{Methods}
We choose the parameter space $K$ of $\bm{\kappa}$ and $\eta$ as
\begin{align}
  K &= (0,+\infty)\times(-\infty,+\infty),\\
  \eta(\bm{\kappa};\bm{x}) &= \kappa_{\ind{0}} x_{\ind{0}} + \kappa_{\ind{1}} x_{\ind{1}} - 2 \nu \log (x_{\ind{1}}+1), \label{eq:d2qIsing_eta}
\end{align}
where $\nu$ is an appropriate positive constant (see below).

This choice of $\eta$ is particularly convenient
for quantum systems which undergo a first-order phase transition
because, as shown in Appendix~\ref{sec:d2qIsing_method},
the density matrix can be approximated as the product of local operators,
and therefore the SE can be numerically constructed easily even for quantum systems.

Since the second order derivative of $\eta$ with respect to $x_{\ind{1}}$
is $2\nu/(x_{\ind{1}}+1)^2$,
the strength of the convexity of $\eta$ can be freely adjusted by tuning the value of $\nu$.
That is, by taking sufficiently large $\nu$,
we can make the order parameter to have a definite value
even in the first-order phase transition region.

From Eq.~\eqref{eq:tdforce-formula}, 
the inverse temperature $\beta = \Pi_{\ind{0}}$
is given directly by $\kappa_{\ind{0}}$ as
\begin{align}
  \Pi_{\indind{N}{0}}^\eta = \kappa_{\ind{0}}. \label{eq:d2qIsing_Pi0}
\end{align}
Therefore, we can investigate properties of the equilibrium states 
of inverse temperature $\beta$ by taking $\kappa_{\ind{0}}=\beta$.
Furthermore, again using Eq.~\eqref{eq:tdforce-formula},
the intensive parameter conjugate to $X_{\ind{1}}$
can also be easily calculated from the expectation value $x_{\indind{N}{1}}^\eta$ as
\begin{align}
  \Pi_{\indind{N}{1}}^\eta = \kappa_{\ind{1}} - \frac{2\nu}{x_{\indind{N}{1}}^\eta+1}. \label{eq:d2qIsing_Pi1}
\end{align}
Therefore, by combining this with Eq.~\eqref{eq:d2qIsing_Pi0},
the magnetic field conjugate to the order parameter can be calculated as
\begin{align}
  f_{\indind{N}{1}}^\eta 
  \equiv - \frac{
    \kappa_{\ind{1}} - \frac{\displaystyle 2\nu}{\displaystyle x_{\indind{N}{1}}^\eta+1}
  }{
    \kappa_{\ind{0}}}. \label{eq:d2qIsing_f1
  }
\end{align}

We calculate statistical-mechanical quantities in the SE
using 
the minimally entangled typical thermal states (METTS) algorithm \cite{White2009,Stoudenmire2010}, 
by extending it to the SE as follows.
We take the quantization axis along $z$-direction.
Then the ``classical product states''
(product of spin eigenstates along the quantization axis) is
\begin{align}
  \ket{s} \equiv \bigotimes_i \ket{s_i}
  \quad (\hat{\sigma}_i^z \ket{s_i} = s_i \ket{s_i}).
\end{align}
We define the METTS of the SE as
\begin{align}
  \ket{s;\eta} \equiv \frac{1}{\sqrt{\braket{s|e^{-N\eta(\hat{\bm{x}}_N)}|s}}} e^{-\frac{1}{2}N\eta(\hat{\bm{x}}_N)} \ket{s}.
\end{align}
Using the completeness of the classical product states,
we can decompose $\hat{\rho}_N^\eta$ into a convex mixture of $\ket{s;\eta}$ as
\begin{align}
  \hat{\rho}_N^\eta
  \propto \sum_{s_i=\pm 1} p(s;\eta) \ket{s;\eta}\bra{s;\eta},
\end{align}
where the weight $p(s;\eta)$ is given by
\begin{align}
  p(s;\eta) \equiv \braket{s|e^{-N\eta(\hat{\bm{x}}_N)}|s}.
\end{align}
Therefore, the expectation value in $\hat{\rho}_N^\eta$ can be calculated
by averaging the expectation value in $\ket{s;\eta}$,
which is sampled according to the weight $p(s;\eta)$.

The numerical simulations are performed
for the square lattice of size $L_1 \times L_2$ with
periodic boundary conditions along $L_1$-direction
and open boundary conditions along $L_2$-direction.

\subsubsection{Results}
Figure~\ref{fig:d2qIsing_histogram} shows
the histogram of the number of samples
as a function of $\braket{s;\eta|\hat{x}_{\indind{N}{1}}|s;\eta}$
for the {\CE} and the SE.
In order to compare them in the first-order phase transition region,
where $\beta \gtrsim 2.15^{-1}$ and $f_{\ind{1}}=0$ \cite{Nagai1987,Nakamura2003},
we focus on the equilibrium state with $\beta=0.5$ and $f_{\ind{1}}=0$.
Hence, we take $\kappa_{\ind{0}}=0.5$
so that $\beta_N^\eta=0.5$
from Eq.~\eqref{eq:d2qIsing_Pi0}.
Furthermore, we take $\kappa_{\ind{1}}=2\nu$.
Then, from Eq.~\eqref{eq:d2qIsing_Pi1},
the SE gives the density matrix
where the order parameter $x_{\ind{1}} \to 0$ in the thermodynamic limit.
That is, we study the equilibrium state with 
$\beta=0.5$ and $x_{\ind{1}}=0$.

For the {\CE}, the histogram of 
Fig.~\ref{fig:d2qIsing_histogram} shows a double peak structure.
This implies that
the density matrix given by the {\CE} is
a statistical mixture of two single-phase states
which have the same values of the intensive parameters.
To see this more directly, we examined order parameter profiles.
Considering that the order parameter $\hat{X}_{\indind{N}{i}}$
is the sum of $\hat{\sigma}_i^z$,
we have calculated the expectation value of $\hat{\sigma}_i^z$
of individual lattice sites
in typical samples of METTS
and have displayed them for two samples 
in Fig.~\ref{fig:d2qIsing_snapshots} (left).
It is seen that
typical METTS of the {\CE} have 
macroscopically homogeneous profiles of the order parameter,
with positive or negative $x_{\ind{1}}$.

In the SE, by contrast, Fig.~\ref{fig:d2qIsing_histogram} shows
that the order parameter has a macroscopically definite value.
This implies that the SE successfully gives a phase coexistence state.
In fact, as shown in Fig.~\ref{fig:d2qIsing_snapshots} (right),
two phases coexist 
in both samples of the SE.
In addition,
since the order parameter profile is different between these samples, 
it is seen that 
the SE is a convex mixture of the states which are macroscopically identical
except for the spatial arrangement of the coexisting phases.
As known from thermodynamics,
the spatial arrangement of the coexisting phases
is not determined by the values of the additive quantities (of the total system) alone. 
Therefore, the SE is consistent with thermodynamics.

\begin{figure}
  \centering
  \includegraphics[keepaspectratio, width=\linewidth]{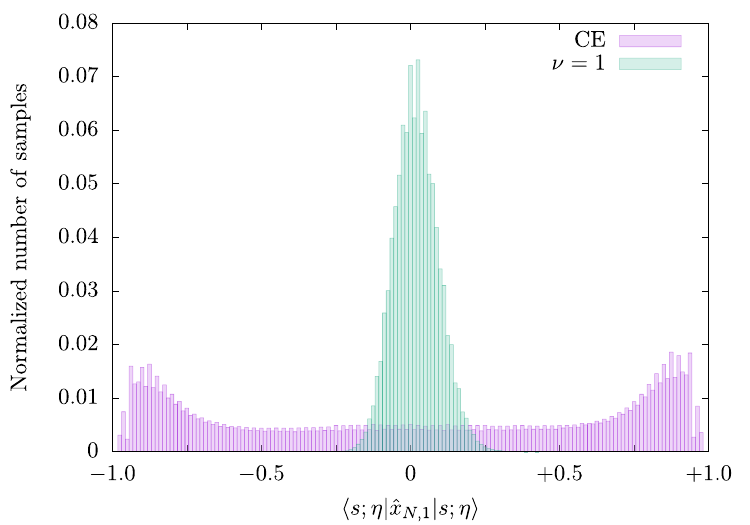}
  \caption{Histogram of the number of samples
  as a function of $\braket{s;\eta|\hat{x}_{\indind{N}{1}}|s;\eta}$
  at the first-order phase transition point ($\beta=0.5$ and $f_{\ind{1}}=0$)
  of the two-dimensional transverse field Ising model ($g=1$) with $L_1=5,L_2=16$
  for the {\CE} and the SE associated with $\eta$ in Eq.~\eqref{eq:d2qIsing_eta}.}
  \label{fig:d2qIsing_histogram}
\end{figure}

\begin{figure}
  \centering
  \includegraphics[keepaspectratio, width=\linewidth]{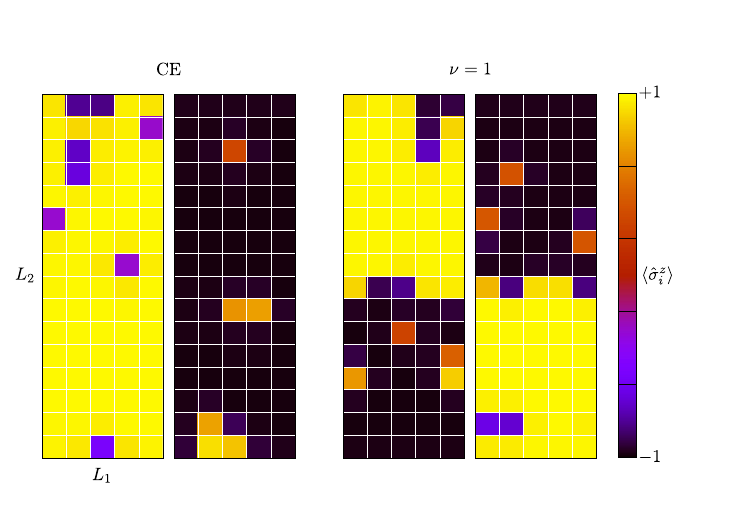}
  \caption{Order parameter profiles in typical 
  samples of METTS of the {\CE} (left)
  and the SE (right) associated with $\eta$ in Eq.~\eqref{eq:d2qIsing_eta}
  at the first-order phase transition point ($\beta=0.5$ and $f_{\ind{1}}=0$)
  of the two-dimensional transverse field Ising model ($g=1$) with $L_1=5,L_2=16$
  (color map, where yellow represents $\braket{\hat{\sigma}}_i^z=+1$
  and black represents $\braket{\hat{\sigma}}_i^z=-1$).
  Two samples   are shown for each ensemble.}
  \label{fig:d2qIsing_snapshots}
\end{figure}

Figure~\ref{fig:d2qIsing_ordered_tdf-op} shows the relation
between the magnetic field and the order parameter.
In the SE,
it is seen that 
$x_{\indind{N}{1}}^\eta$ is a multivalued function of $f_{\indind{N}{1}}^\eta$,
as in the classical case $g = 0$, Fig.~\ref{fig:d2cIsing_ordered_tdf-op}.
(Note that only the upper half region,
in which  $x_{\indind{N}{1}}^\eta \geq 0$, 
is plotted in Fig.~\ref{fig:d2qIsing_ordered_tdf-op}.)
This behavior is due to the presence of interfaces (domain walls)
separating the coexisting phases,
which are clearly observed in Fig.~\ref{fig:d2qIsing_snapshots} (right) \cite{Challa1988_PRA,Yoneta2019}.
In the {\CE}, by contrast, it is a monotonous and continuous function.

\begin{figure}
  \centering
  \includegraphics[keepaspectratio, width=\linewidth]{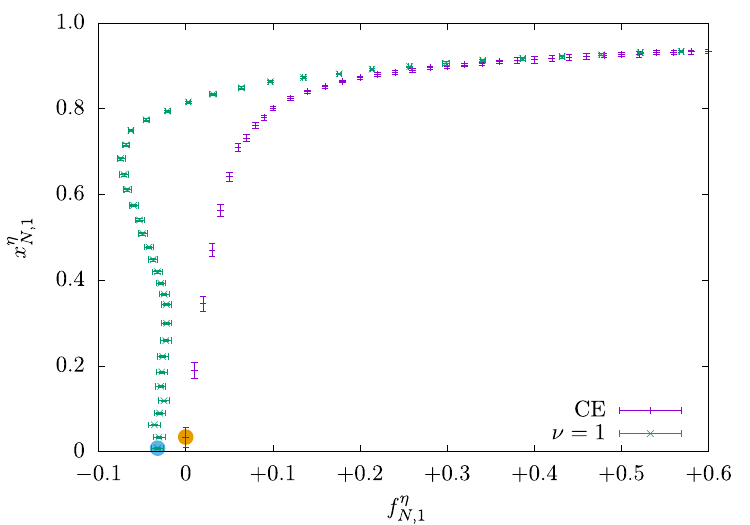}
  \caption{Relation between the magnetic field and the order parameter
  in the vicinity of the first-order phase transition point ($\beta=0.5$ and $f_{\ind{1}}=0$)
  of the two-dimensional transverse field Ising model ($g=1$) with $L_1=5,L_2=16$
  for the {\CE} and the SE associated with $\eta$ in Eq.~\eqref{eq:d2qIsing_eta}.
  The filled circles indicate the points corresponding to the parameters of Figs.~\ref{fig:d2qIsing_histogram} and \ref{fig:d2qIsing_snapshots}
  for the {\CE} and the SE, respectively, 
  where slight deviations from $(0,0)$ are due to finite-size effects.}
  \label{fig:d2qIsing_ordered_tdf-op}
\end{figure}

These results demonstrate that
the SE successfully gives a phase coexistence state
with phase interfaces,
whereas the {\CE} gives a convex mixture of single-phase states.

\section{Summary and Discussion} \label{sec:summary}
We have extended a generalized ensemble,
which we call the squeezed ensemble (SE),
in such a way that it is applicable to equilibrium states
specified by noncommutative additive observables.

In contrast to the canonical ensemble,
additive observables which specify the equilibrium state
always have macroscopically definite values in the SE.
Furthermore, 
unlike the microcanonical and the restricted ensembles, the SE
is well-defined and well-behaved even when the order parameter(s)
does not commute with the other additive observables
such as the Hamiltonian.
Therefore, it solves the fundamental problems
faced by conventional ensembles
in studying such quantum systems when they exhibit first-order phase transitions.

Various ensembles are included in the class of the SE.
One can choose an appropriate SE depending on the purpose,
such as one that is easy to construct.
In addition, good analytic properties of the the SE
yield practical formulas for thermodynamic quantities
such as the entropy and the intensive parameter.
For these reasons,
the SE is convenient for practical calculations.
We have thus established, for the first time,
a concrete method for constructing phase coexistence states
of general quantum systems.

As an demonstration, we have applied our formulation
to the two-dimensional transverse field Ising model at finite temperature,
whose order parameter does not commute with the Hamiltonian.
We have confirmed that the SE is composed of thermodynamically identical states
with large domains (ordered phases) separated by interfaces (domain walls).
In addition, we have obtained thermodynamic anomalies
peculiar to such phase coexistence states in finite systems.
These results demonstrate that the SE successfully gives a phase coexistence state.
From a viewpoint of applications to engineering, 
we expect that the SE enables the incorporation of quantum effects
into the analysis and design of materials with phase interfaces.

In this paper, we have formulated the statistical mechanics
only for the leading terms of statistical mechanical quantities
in the thermodynamic limit
for the case where the additive observables do not commute with each other.
However, finite-size effects due to phase interfaces,
such as the negative susceptibility,
are also important properties of phase coexistence states.
From the numerical results,
for the equilibrium state without phase interfaces,
it is expected that finite-size effects
on statistical-mechanical quantities of $O(N^0)$ given by the SE
is generally $O(N^{-1})$.
Therefore, the ensemble dependence,
which is also a finite-size effect, should be $O(N^{-1})$.
Now let us assume that this also holds for equilibrium states with phase interfaces.
Then, since the spatial dimension $\nu$ is larger than $1$
for systems that exhibit first-order phase transitions,
the finite-size effects of $O(N^{-1/\nu})$,
which includes the effects of the phase interfaces,
are independent of the ensemble and are expected to be physical.
Therefore, it is expected that the SE can correctly give the finite-size effects due to the phase interfaces.
A more detailed analysis of the finite-size effects of the SE is left for future work.

\begin{acknowledgments}
We thank Y. Chiba, H. Tasaki, and H. Hakoshima for discussions.
This work was supported by The Japan Society for the Promotion of Science, KAKENHI No.~19H01810, No.~20J22982 and No.~22H01142.
\end{acknowledgments}

\onecolumngrid
\appendix

\section{Derivation of Eq.~\eqref{eq:freeSpin_RE_var}} \label{sec:derivation_freeSpin_RE}
We first calculate the ``partition function'' of the RE.
The projection operator $\hat{P}_N$ can be written as
\begin{align}
  \hat{P}_N = \sum_{M=M^-}^{M^+} \sum_{\sum s_i=M} \bigotimes_{i=1}^N \ket{s_i}\bra{s_i}
  \quad (\hat{\sigma}_i^y \ket{s_i} = s_i \ket{s_i}),
\end{align}
where $M^+ \equiv 2 \lfloor N m^+ / 2 \rfloor, M^- \equiv 2 \lceil N m^- / 2 \rceil$.
Therefore, we have
\begin{align}
  \Tr \left[ \hat{P}_N e^{-\pi_{\ind{1}}\hat{X}_{\indind{N}{1}}} \hat{P}_N \right]
  &= 2^N \cosh^N\pi_{\ind{1}} \sum_{M=M^-}^{M^+} \binom{N}{\frac{N+M}{2}}.
\end{align}

Next, we evaluate the expectation value of $\hat{x}_{\indind{N}{1}}$ in $\hat{\rho}_N^\mathrm{r}$.
Without loss of generality,
we assume $\pi_{\ind{1}} < 0$.
Then
\begin{align}
  (0<)
  &\Tr \left[ \sigma_i^x \hat{P}_N e^{-\pi_{\ind{1}}\hat{X}_{\indind{N}{1}}} \hat{P}_N \right] \nonumber\\
  &= - 2^N \sinh\pi_{\ind{1}} \cosh^{N-1}\pi_{\ind{1}} \left(
    \sum_{M=M^-+2}^{M^+-2} \binom{N}{\frac{N+M}{2}}
    + \binom{N-1}{\frac{(N-1)+(M^+-1)}{2}}
    + \binom{N-1}{\frac{(N-1)+(M^-+1)}{2}}
  \right) \nonumber\\
  &= -2^N  \sinh\pi_{\ind{1}} \cosh^{N-1}\pi_{\ind{1}} \left(
    \sum_{M=M^-}^{M^+} \binom{N}{\frac{N+M}{2}}
    - \frac{N-M^+}{2N} \binom{N}{\frac{N+M^+}{2}}
    - \frac{N+M^-}{2N} \binom{N}{\frac{N+M^-}{2}}
  \right) \nonumber\\
  &< - 2^N \sinh\pi_{\ind{1}} \cosh^{N-1}\pi_{\ind{1}} \left(
    1 - \frac{\frac{N+M^-}{2N} \binom{N}{\frac{N+M^-}{2}}}{\sum_{M=M^-}^{M^+} \binom{N}{\frac{N+M}{2}}}
  \right) \sum_{M=M^-}^{M^+} \binom{N}{\frac{N+M}{2}} \nonumber\\
  &< - 2^N \sinh\pi_{\ind{1}} \cosh^{N-1}\pi_{\ind{1}} (1 - M^-/N) \sum_{M=M^-}^{M^+} \binom{N}{\frac{N+M}{2}}
\end{align}
In the last inequality
we have used
\begin{align}
  \sum_{M=M^-}^{M^+} \binom{N}{\frac{N+M}{2}}
  &< \binom{N}{\frac{N+M^-}{2}} \sum_{n=0}^{M^+-M^-} \left( \frac{N-M^-}{N+M^-} \right)^n
  = \binom{N}{\frac{N+M^-}{2}} \left\{ \frac{1}{1-\frac{N-M^-}{N+M^-}} - e^{-\Theta(N)} \right\}.
\end{align}
Therefore, using $\displaystyle \lim_{N\to\infty} M^\pm/N = m^\pm$, we have
\begin{align}
  (0\leq) \lim_{N\to\infty} \Tr \left[ \hat{x}_{\indind{N}{1}} \hat{\rho}_N^\mathrm{r} \right]
  &\leq - (1-m^-) \tanh\pi_{\ind{1}}.
\end{align}

Finally, we evaluate the expectation value of $\left(\hat{x}_{\indind{N}{1}}\right)^2$ in $\hat{\rho}_N^\mathrm{r}$.
In the similar manner as above, for any $0<\epsilon<m^+-m^-$, we have
\begin{align}
  &\Tr \left[ \sigma_i^x \sigma_j^x \hat{P}_N e^{-\pi_{\ind{1}}\hat{X}_{\indind{N}{1}}} \hat{P}_N \right] \nonumber\\
  &= 2^N \sinh^2\pi_{\ind{1}} \cosh^{N-2}\pi_{\ind{1}} \left(
    2 \sum_{M=M^-}^{M^+} \binom{N-2}{\frac{(N-2)+M}{2}}
    + \sum_{M=M^-+4}^{M^+} \binom{N-2}{\frac{(N-2)+(M-2)}{2}}
    + \sum_{M=M^-}^{M^+-4} \binom{N-2}{\frac{(N-2)+(M+2)}{2}}
  \right) \nonumber\\
  &= 2^N \sinh^2\pi_{\ind{1}} \cosh^{N-2}\pi_{\ind{1}} \left(
    \sum_{M=M^-}^{M^+} \frac{(N+M)(N-M)}{2N(N-1)} \binom{N}{\frac{N+M}{2}}
    + \sum_{M=M^-}^{M^+-4} \frac{(N-M)(N-M-2)}{2N(N-1)} \binom{N}{\frac{N+M}{2}}
  \right) \nonumber\\
  &> 2^N \sinh^2\pi_{\ind{1}} \cosh^{N-2}\pi_{\ind{1}} \sum_{M=M^-}^{M^-+\lfloor N\epsilon \rfloor} \frac{(N+M)(N-M)+(N-M)(N-M-2)}{2N(N-1)} \binom{N}{\frac{N+M}{2}} \nonumber\\
  &= 2^N \sinh^2\pi_{\ind{1}} \cosh^{N-2}\pi_{\ind{1}} \sum_{M=M^-}^{M^-+\lfloor N\epsilon \rfloor} (1-M/N) \binom{N}{\frac{N+M}{2}} \nonumber\\
  &> 2^N \sinh^2\pi_{\ind{1}} \cosh^{N-2}\pi_{\ind{1}} (1-M^-/N-\epsilon) \sum_{M=M^-}^{M^-+\lfloor N\epsilon \rfloor} \binom{N}{\frac{N+M}{2}} \nonumber\\
  &= 2^N \sinh^2\pi_{\ind{1}} \cosh^{N-2}\pi_{\ind{1}} (1-M^-/N-\epsilon-e^{-\Theta(N)}) \sum_{M=M^-}^{M^+} \binom{N}{\frac{N+M}{2}}.
\end{align}
Therefore, we have
\begin{align}
  \lim_{N\to\infty} \Tr \left[ \left(\hat{x}_{\indind{N}{1}}\right)^2 \hat{\rho}_N^\mathrm{r} \right]
  &\geq (1-m^--\epsilon) \tanh^2\pi_{\ind{1}}.
\end{align}
Since the choice of $\epsilon$ was arbitrary, passing to the limit as $\epsilon \to 0$, we obtain
\begin{align}
  \lim_{N\to\infty} \Tr \left[ \left(\hat{x}_{\indind{N}{1}}\right)^2 \hat{\rho}_N^\mathrm{r} \right]
  &\geq (1-m^-) \tanh^2\pi_{\ind{1}}.
\end{align}

\section{Proof of Theorem~\ref{theorem:psi}} \label{sec:proof_theorem_psi}
Using Eq.~\eqref{eq:hat-check} and $\|\hat{x}_{\indind{N}{i}}\|\leq\|\hat{o}_{\ind{i}}\|(<\infty)$,
we have
\begin{align}
  \lim_{N\to\infty}\left\| \eta(\hat{\bm{x}}_N) - \eta(\check{\bm{x}}_N) \right\|=0. \label{eq:hat-check_eta}
\end{align}
That is, for any $\epsilon > 0$, there exists $N_0\in\mathbb{N}$ such that
$\left\| \eta(\hat{\bm{x}}_N) - \eta(\check{\bm{x}}_N) \right\|<\epsilon$ for any $N>N_0$.
Then it follows that
\begin{align}
  - \eta(\check{\bm{x}}_N) - \epsilon < - \eta(\hat{\bm{x}}_N) < - \eta(\check{\bm{x}}_N) + \epsilon.
\end{align}
Therefore, by using Weyl's inequality \cite{Bhatia1997}, we have
\begin{align}
  \Tr \left[ e^{- N \eta(\check{\bm{x}}_N) - N \epsilon} \right]
  < \Tr \left[ e^{- N \eta(\hat{\bm{x}}_N)} \right]
  < \Tr \left[ e^{- N \eta(\check{\bm{x}}_N) + N \epsilon} \right].
\end{align}
By definition of $\psi_N^{\dot{\eta}}$, this implies that
\begin{align}
  \psi_N^{\check{\eta}} - \epsilon < \psi_N^{\hat{\eta}} < \psi_N^{\check{\eta}} + \epsilon.
\end{align}

\section{Proof of Theorem~\ref{theorem:exp}} \label{sec:proof_theorem_exp}
To prove Theorem~\ref{theorem:exp}, we first prove the following lemma.
\begin{lemma} \label{lemma:AA_cor}
Let $\{f_n\}_{n \in \mathbb{N}}$ be a sequence of real-valued differentiable functions on $I$
such that $\{f_n(x_0)\}_{n \in \mathbb{N}}$ is bounded for some $x_0 \in I$
and $\{f_n'\}_{n \in \mathbb{N}}$ is uniformly bounded.
Then there exists a subsequence $\{f_{n_k}\}_{k \in \mathbb{N}}$ that converges uniformly on $I$.
\end{lemma}
\begin{proof}
By the Arzel\`{a}-Ascoli theorem \cite{Rudin1976},
it is sufficient to show
the uniform boundedness and equicontinuity
of $\{f_n\}_{n \in \mathbb{N}}$.

(uniform boundedness)
Uniform boundedness of $\{f_n'\}_{n \in \mathbb{N}}$ implies by the mean value theorem that for all $n \in \mathbb{N}$ and $x,y \in I$,
\begin{align}
  \left| f_n(x)-f_n(y) \right| \leq M \left| x-y \right|, \label{eq:AA_Cor_MeanValueThm}
\end{align}
where $\displaystyle M \equiv \sup_{n \in \mathbb{N}, x \in I}|f_n'(x)|$.
In addition, since $\{f_n(x_0)\}_{n \in \mathbb{N}}$ is bounded,
we have, for all $n \in \mathbb{N}$ and $x \in I = [a,b]$,
\begin{align}
  \left| f_n(x) \right| \leq F + M |b-a|,
\end{align}
where $\displaystyle F \equiv \sup_{n \in \mathbb{N}}|f_n(x_0)|$.
Therefore, $\{f_n\}_{n \in \mathbb{N}}$ is uniformly bounded.

(equicontinuity)
Given $\epsilon > 0$, let $\displaystyle \delta \equiv \frac{\epsilon}{M}$.
Then, using Eq.~\eqref{eq:AA_Cor_MeanValueThm},
we have
\begin{align}
  \left| x-y \right| < \delta
  \Rightarrow \left| f_n(x)-f_n(y) \right| < \epsilon
\end{align}
for all $x, y \in I$.
Therefore, $\{f_n\}_{n \in \mathbb{N}}$ is equicontinuous.
\end{proof}

\begin{proof}[Proof of Theorem~\ref{theorem:exp}]
Let us introduce the generating functions
\begin{align}
  \phi_N^{\dot{\eta},\dot{\theta}}(\lambda)
  &\equiv - \frac{1}{N} \log \Tr \left[ e^{-N\left(\eta(\dot{\bm{x}}_N)+\lambda\theta(\dot{\bm{x}}_N)\right)} \right].
\end{align}
Taking the derivative of $\phi_N^{\dot{\eta},\dot{\theta}}$,
we obtain the expectation value of $\theta(\dot{\bm{x}}_N)$ in $\rho_N^{\dot{\eta},\dot{\theta}}(\lambda)$ as
\begin{align}
  \frac{d\phi_N^{\dot{\eta},\dot{\theta}}}{d\lambda}(\lambda)
  &= \braket{\theta(\dot{\bm{x}}_N)}_N^{\dot{\eta},\dot{\theta}}(\lambda).
\end{align}
In particular, at $\lambda=0$, it coincides with the expectation value in $\rho_N^{\dot{\eta}}$.

Let $d$ be the dimension of the local Hilbert space on each site.
Then we have
\begin{align}
  d^N e^{-N\|\eta(\check{\bm{x}}_N)\|} \leq \Tr \left[ e^{-N\eta(\check{\bm{x}}_N)} \right] \leq d^N e^{+N\|\eta(\check{\bm{x}}_N)\|}.
\end{align}
Therefore, we observe that $\left\{ \phi_N^{\check{\eta},\check{\theta}} (0) \right\}_{N\in\mathbb{N}}$ is bounded.
In addition, since $\displaystyle\frac{d\phi_N^{\check{\eta},\check{\theta}}}{d\lambda}(\lambda)=\braket{\theta(\check{\bm{x}}_N)}_N^{\check{\eta},\check{\theta}}(\lambda)\leq\left\|\theta(\check{\bm{x}}_N)\right\|$,
$\displaystyle\left\{\frac{d\phi_N^{\check{\eta},\check{\theta}}}{d\lambda}\right\}_{N\in\mathbb{N}}$ is uniformly bounded on any closed and bounded interval.
Therefore, by Lemma~\ref{lemma:AA_cor},
there exists a subsequence $\displaystyle\left\{\phi_{N_k}^{\check{\eta},\check{\theta}}\right\}_{k\in\mathbb{N}}$
that converges (uniformly on compacts in $\mathbb{R}$).
For this limit, in the same way as in the proof of Theorem~\ref{theorem:psi},
it can be shown that
\begin{align}
  \phi^{\eta,\theta}(\lambda)
  \equiv \lim_{k\to\infty}\phi_{N_k}^{\hat{\eta},\hat{\theta}}(\lambda)
  = \lim_{k\to\infty}\phi_{N_k}^{\check{\eta},\check{\theta}}(\lambda)
\end{align}
for all $\lambda \in \mathbb{R}$.

Since $\theta(\check{\bm{x}}_N)$ commutes with $\eta(\check{\bm{x}}_N)$,
we have
\begin{align}
  \frac{d^2\phi_N^{\check{\eta},\check{\theta}}}{d\lambda^2}(\lambda)
  &= - N \Tr \left[ \left( \theta(\check{\bm{x}}_N) - \braket{\theta(\check{\bm{x}}_N)}_N^{\check{\eta},\check{\theta}}(\lambda) \right)^2 \rho_N^{\check{\eta},\check{\theta}}(\lambda) \right].
\end{align}
Thus, using assumption~\ref{cond:variance},
we find that $\displaystyle\left\{\frac{d^2\phi_N^{\check{\eta},\check{\theta}}}{d\lambda^2}\right\}_{N\in\mathbb{N}}$ is uniformly bounded on $I$.
Hence, by Lemma~\ref{lemma:AA_cor},
there exists a subsequence $\displaystyle\left\{\frac{d\phi_{N_{k_l}}^{\check{\eta},\check{\theta}}}{d\lambda}\right\}_{l\in\mathbb{N}}$ that converges uniformly on $I$.
Therefore, $\phi^{\eta,\theta}$ is continuously differentiable on $I$
and its derivative is given by
\begin{align}
  \frac{d\phi^{\eta,\theta}}{d\lambda}(\lambda)
  = \lim_{l\to\infty}\frac{d\phi_{N_{k_l}}^{\check{\eta},\check{\theta}}}{d\lambda}(\lambda)
  = \lim_{l\to\infty}\braket{\theta(\check{\bm{x}}_{N_{k_l}})}_{N_{k_l}}^{\check{\eta},\check{\theta}}(\lambda)
\end{align}
for all $\lambda \in I$ \cite{Rudin1976}.

Since $\phi_{N_k}^{\hat{\eta},\hat{\theta}}$ is concave,
using Griffiths's lemma\cite{Griffiths1964}, we obtain
\begin{align}
  \frac{d\phi^{\eta,\theta}}{d\lambda}(\lambda)
  = \lim_{k\to\infty}\frac{d\phi_{N_k}^{\hat{\eta},\hat{\theta}}}{d\lambda}(\lambda)
  = \lim_{k\to\infty}\braket{\theta(\hat{\bm{x}}_{N_k})}_{N_k}^{\hat{\eta},\hat{\theta}}(\lambda).
\end{align}
for all $\lambda \in I$.

From the above, together with assumption~\ref{cond:exp_converge},
we have
\begin{align}
  \lim_{N\to\infty}\braket{\theta(\check{\bm{x}}_N)}_N^{\check{\eta},\check{\theta}}(0)
  = \lim_{l\to\infty}\braket{\theta(\check{\bm{x}}_{N_{k_l}})}_{N_{k_l}}^{\check{\eta},\check{\theta}}(0)
  = \lim_{k\to\infty}\braket{\theta(\hat{\bm{x}}_{N_k})}_{N_k}^{\hat{\eta},\hat{\theta}}(0)
  = \lim_{N\to\infty}\braket{\theta(\hat{\bm{x}}_N)}_N^{\hat{\eta},\hat{\theta}}(0).
\end{align}
\end{proof}

\section{Proof of Properties~\ref{property:x_N^eta}-\ref{formula:tdforce}} \label{sec:properties_proof}
\begin{proof}[proof of Property~\ref{property:x_N^eta}]
Using Eq.~\eqref{eq:tr_f-rho_check} and letting $f(\bm{x})=x_i$, we get
\begin{align}
  \lim_{N\to\infty} \Braket{\check{x}_{\indind{N}{i}}}_N^{\check{\eta}}
  = x_{\indind{\mathrm{max}}{i}}^\eta.
\end{align}
We then apply Theorem~\ref{theorem:exp} with $\theta=x_{\ind{i}}$,
and thereby relate this to the expectation value of $\hat{x}_{\indind{N}{i}}$
in the SE for the noncommutative set $\hat{\bm{x}}_N$.
To do so, we check the assumptions of the theorem.
More specifically, we examine the variance of $\theta(\check{x}_{\indind{N}{i}})=\check{x}_{\indind{N}{i}}$ in the density matrix
\begin{align}
  \rho_N^{\check{\eta},\check{\theta}}(\lambda) \propto e^{-N(\eta(\check{\bm{x}}_N)+\lambda\theta(\check{\bm{x}}_N))},
\end{align}
which is constructed from the commutative set $\check{\bm{x}}_N$.

First, we show that there exists a bounded closed interval $I(\ni 0)$
such that for any $\lambda \in I$, $\eta+\lambda\theta$ satisfies condition~\ref{cond:B}.
Since we have chosen $\theta=x_{\ind{i}}$,
the Hesse matrix $H[\s-(\eta+\lambda\theta)](\bm{x})$ of $\s(\bm{x})-(\eta(\bm{x})+\lambda\theta(\bm{x}))$
is equal to the Hesse matrix $H[\s-\eta](\bm{x})$ of $\s(\bm{x})-\eta(\bm{x})$:
\begin{align}
  &H_{kl}[\s-(\eta+\lambda\theta)](\bm{x})
  = H_{kl}[\s-\eta](\bm{x}) \nonumber\\
  &= \frac{\partial^2 \s}{\partial x_{\ind{k}} \partial x_{\ind{l}}}(\bm{x})
  - \frac{\partial^2 \eta}{\partial x_{\ind{k}} \partial x_{\ind{l}}}(\bm{x})
\end{align}
Since $\eta$ satisfies condition~\ref{cond:B},
the Hesse matrix of $\s-\eta$ is negative definite
in a neighborhood of $\bm{x}_\mathrm{max}^\eta$.
Therefore,
the Hesse matrix of $\s(\bm{x})-(\eta(\bm{x})+\lambda\theta(\bm{x}))$ is also negative definite
in a neighborhood of $\bm{x}_\mathrm{max}^\eta$.
By the implicit function theorem, in a neighborhood of $\lambda=0$,
the maximum point $\bm{x}_\mathrm{max}^{\eta,\theta}(\lambda)$ of $\s(\bm{x})-(\eta(\bm{x})+\lambda\theta(\bm{x}))$
changes continuously from $\bm{x}_\mathrm{max}^\eta$ with respect to $\lambda$.
Therefore, taking $I(\ni 0)$ narrow enough, for any $\lambda \in I$,
the Hesse matrix of $\s(\bm{x})-(\eta(\bm{x})+\lambda\theta(\bm{x}))$
is negative definite in a neighborhood of $\bm{x}_\mathrm{max}^{\eta,\theta}(\lambda)$.
That is, for any $\lambda \in I$, $\eta+\lambda\theta$ satisfies condition~\ref{cond:B}.

Next, for $I$ taken as above,
we show that Eq.~\eqref{eq:variance_cond} of assumption~\ref{cond:variance} holds.
Let $H$ be a Hesse Matrix of $\s(\bm{x})-(\eta(\bm{x})+\lambda\theta(\bm{x}))$ at $\bm{x}_\mathrm{max}^{\eta,\theta}(\lambda)$.
Since $\eta+\lambda\theta$ also satisfies condition~\ref{cond:B},
in the same way as discussed in Section~\ref{sec:SE_noncommutative_check},
it can be seen that
the probability distribution of $\check{\bm{x}}_N$ in the density matrix $\rho_N^{\check{\eta},\check{\theta}}(\lambda)$ can be well approximated by the Gaussian distribution with the covariance matrix $- N^{-1} H^{-1}$.
Therefore, we obtain
\begin{align}
  &\Tr \left[ \left( \check{x}_{\indind{N}{i}} - \braket{\check{x}_{\indind{N}{i}}}_N^{\check{\eta},\check{\theta}}(\lambda) \right)^2 \rho_N^{\check{\eta},\check{\theta}}(\lambda) \right] \nonumber\\
  &= N^{-1} |(H^{-1})_{ii}|  + o(N^{-1}).
\end{align}
Thus, we see that indeed assumption~\ref{cond:variance} of Theorem~\ref{theorem:exp} holds.

From the above, we can apply Theorem~\ref{theorem:exp} and we obtain
\begin{align}
  \lim_{N\to\infty} \braket{\hat{x}_{\indind{N}{i}}}_N^{\hat{\eta}}
  = \lim_{N\to\infty} \braket{\check{x}_{\indind{N}{i}}}_N^{\check{\eta}}
  = x_{\indind{\mathrm{max}}{i}}^\eta. \label{eq:asymptotic-approx_x}
\end{align}
Therefore, we have Eq.~\eqref{eq:x_N^eta}.
\end{proof}

\begin{proof}[proof of Property~\ref{property:var_x_N^eta}]
Using Eq.~\eqref{eq:tr_f-rho_check} and letting $f(\bm{x})=(x_i)^2$, we get
\begin{align}
  \lim_{N\to\infty} \Braket{(\check{x}_{\indind{N}{i}})^2}_N^{\check{\eta}}
  = (x_{\indind{\mathrm{max}}{i}}^\eta)^2.
\end{align}
We then apply Theorem~\ref{theorem:exp} with $\theta=(x_{\ind{i}})^2$,
and thereby relate this to the expectation value of $(\hat{x}_{\indind{N}{i}})^2$
in the SE for the noncommutative set $\hat{\bm{x}}_N$.
To do so, we check the assumptions of the theorem.
More specifically, we examine the variance of $\theta(\check{x}_{\indind{N}{i}})=(\check{x}_{\indind{N}{i}})^2$ in the density matrix
\begin{align}
  \rho_N^{\check{\eta},\check{\theta}}(\lambda) \propto e^{-N(\eta(\check{\bm{x}}_N)+\lambda\theta(\check{\bm{x}}_N))},
\end{align}
which is constructed from the commutative set $\check{\bm{x}}_N$.

First, we show that there exists a bounded closed interval $I(\ni 0)$
such that for any $\lambda \in I$, $\eta+\lambda\theta$ satisfies condition~\ref{cond:B}.
Since we have chosen $\theta=(x_{\ind{i}})^2$,
the Hesse matrix $H[\s-(\eta+\lambda\theta)](\bm{x})$ of $\s(\bm{x})-(\eta(\bm{x})+\lambda\theta(\bm{x}))$
is related to the Hesse matrix $H[\s-\eta](\bm{x})$ of $\s(\bm{x})-\eta(\bm{x})$ as
\begin{align}
  H_{kl}[\s-(\eta+\lambda\theta)](\bm{x})
  = H_{kl}[\s-\eta](\bm{x}) - 2 \lambda \delta_{ik} \delta_{il}.
\end{align}
Thus, when $H[\s-\eta](\bm{x})$ is negative definite,
$H[\s-(\eta+\lambda\theta)](\bm{x})$ is also negative definite.
Since $\eta$ satisfies condition~\ref{cond:B},
the Hesse matrix of $\s-\eta$ is negative definite
in a neighborhood of $\bm{x}_\mathrm{max}^\eta$.
Therefore,
the Hesse matrix of $\s(\bm{x})-(\eta(\bm{x})+\lambda\theta(\bm{x}))$ is also negative definite
in a neighborhood of $\bm{x}_\mathrm{max}^\eta$.
By the implicit function theorem, in a neighborhood of $\lambda=0$,
the maximum point $\bm{x}_\mathrm{max}^{\eta,\theta}(\lambda)$ of $\s(\bm{x})-(\eta(\bm{x})+\lambda\theta(\bm{x}))$
changes continuously from $\bm{x}_\mathrm{max}^\eta$ with respect to $\lambda$.
Therefore, taking $I=[0,\delta]$ for some positive constant $\delta$, for any $\lambda \in I$,
the Hesse matrix of $\s(\bm{x})-(\eta(\bm{x})+\lambda\theta(\bm{x}))$
is negative definite in a neighborhood of $\bm{x}_\mathrm{max}^{\eta,\theta}(\lambda)$.
That is, for any $\lambda \in I$, $\eta+\lambda\theta$ satisfies condition~\ref{cond:B}.

Next, for $I$ taken as above,
we show that Eq.~\eqref{eq:variance_cond} of assumption~\ref{cond:variance} holds.
Let $H$ be a Hesse Matrix of $\s(\bm{x})-(\eta(\bm{x})+\lambda\theta(\bm{x}))$ at $\bm{x}_\mathrm{max}^{\eta,\theta}(\lambda)$.
Since $\eta+\lambda\theta$ also satisfies condition~\ref{cond:B},
in the same way as discussed in Section~\ref{sec:SE_noncommutative_check},
it can be seen that
the probability distribution of $\check{\bm{x}}_N$ in the density matrix $\rho_N^{\check{\eta},\check{\theta}}(\lambda)$ can be well approximated by the Gaussian distribution with the covariance matrix $- N^{-1} H^{-1}$.
Therefore, denoting $\delta\check{x}_{\indind{N}{i}}(\lambda) \equiv \check{x}_{\indind{N}{i}} - \braket{\check{x}_{\indind{N}{i}}}_N^{\check{\eta},\check{\theta}}(\lambda)$, we obtain
\begin{align}
  &\Tr \left[ \left( (\check{x}_{\indind{N}{i}})^2 - \braket{(\check{x}_{\indind{N}{i}})^2}_N^{\check{\eta},\check{\theta}}(\lambda) \right)^2 \rho_N^{\check{\eta},\check{\theta}}(\lambda) \right] \nonumber\\
  &= 4 \left(\braket{\check{x}_{\indind{N}{i}}}_N^{\check{\eta},\check{\theta}}(\lambda)\right)^2 \braket{(\delta\check{x}_{\indind{N}{i}}(\lambda))^2}_N^{\check{\eta},\check{\theta}}(\lambda) \nonumber\\
  &+ 4 \braket{\check{x}_{\indind{N}{i}}}_N^{\check{\eta},\check{\theta}}(\lambda) \braket{(\delta\check{x}_{\indind{N}{i}}(\lambda))^3}_N^{\check{\eta},\check{\theta}}(\lambda) \nonumber\\
  &+ \braket{(\delta\check{x}_{\indind{N}{i}}(\lambda))^4}_N^{\check{\eta},\check{\theta}}(\lambda)
  - \left(\braket{(\delta\check{x}_{\indind{N}{i}}(\lambda))^2}_N^{\check{\eta},\check{\theta}}(\lambda)\right)^2 \nonumber\\
  &= 4 N^{-1} \left(\braket{\check{x}_{\indind{N}{i}}}_N^{\check{\eta},\check{\theta}}(\lambda)\right)^2 |(H^{-1})_{ii}| + o(N^{-1}).
\end{align}
Thus, we see that indeed assumption~\ref{cond:variance} of Theorem~\ref{theorem:exp} holds.

From the above, we can apply Theorem~\ref{theorem:exp} and we obtain
\begin{align}
  \lim_{N\to\infty} \braket{(\hat{x}_{\indind{N}{i}})^2}_N^{\hat{\eta}}
  = \lim_{N\to\infty} \braket{(\check{x}_{\indind{N}{i}})^2}_N^{\check{\eta}}
  = (x_{\indind{\mathrm{max}}{i}}^\eta)^2.
\end{align}
Therefore, together with Eq.~\eqref{eq:asymptotic-approx_x}, we have Eq.~\eqref{eq:var_x_N^eta}.
\end{proof}

\begin{proof}[proof of Property~\ref{formula:entropy}]
By using Theorem~\ref{theorem:psi} and Eq.~\eqref{eq:psi_check}, in the thermodynamic limit, we have
\begin{align}
  \lim_{N\to\infty} \psi_N^{\hat{\eta}}
  = \lim_{N\to\infty} \psi_N^{\check{\eta}}
  &= \eta (\bm{x}_\mathrm{max}^\eta) - \s(\bm{x}_\mathrm{max}^\eta).
\end{align}
This, together with Eq.~\eqref{eq:asymptotic-approx_x}, implies
\begin{align}
  \eta(\bm{x}_N^{\hat{\eta}}) - \psi_N^{\hat{\eta}}
  = \s(\bm{x}_\mathrm{max}^\eta) + o(N^0).
\end{align}
Therefore, we have Eq.~\eqref{eq:entropy-formula}.
\end{proof}

\begin{proof}[proof of Property~\ref{formula:tdforce}]
By using Eq.~\eqref{eq:x_max^eta},
which is also true for the SE for $\check{\bm{x}}_N$ as described in Section~\ref{sec:SE_noncommutative_check},
we have
\begin{align}
  \Pi_{\ind{i}} (\bm{x}_\mathrm{max}^\eta)
  = \frac{\partial \s}{\partial x_{\ind{i}}} (\bm{x}_\mathrm{max}^\eta)
  = \frac{\partial \eta}{\partial x_{\ind{i}}} (\bm{x}_\mathrm{max}^\eta).
\end{align}
This, together with Eq.~\eqref{eq:asymptotic-approx_x}, implies
\begin{align}
  \frac{\partial\eta}{\partial x_{\ind{i}}}(\bm{x}_N^{\hat{\eta}})
  = \Pi_{\ind{i}}(\bm{x}_\mathrm{max}^\eta) + o(N^0).
\end{align}
Therefore, we have Eq.~\eqref{eq:tdforce-formula}.
\end{proof}

\section{Proof of Property~\ref{property:psi-s}} \label{sec:genLeg_proof}
Using Eq.~\eqref{eq:s-psi}, we have
\begin{align}
  \s(\bm{x})
  &\leq \eta(\bm{\kappa};\bm{x}) - \inf_{\bm{x}' \in \Omega} \left\{ \eta(\bm{\kappa};\bm{x}')-\s(\bm{x}') \right\} \nonumber\\
  &= \eta(\bm{\kappa};\bm{x}) - \psi^\eta(\bm{\kappa})
\end{align}
for all $\bm{\kappa} \in K$.
Then we have
\begin{align}
   \s(\bm{x}) \leq \inf_{\bm{\kappa} \in K} \left\{ \eta(\bm{\kappa};\bm{x})-\psi^\eta(\bm{\kappa}) \right\}.
\end{align}
On the other hand, using Eq.~\eqref{eq:s-psi}, we have
\begin{align}
  &\inf_{\bm{\kappa} \in K} \left\{ \eta(\bm{\kappa};\bm{x})-\psi^\eta(\bm{\kappa}) \right\} \nonumber\\
  &\leq \eta(\bm{\kappa}';\bm{x})-\psi^\eta(\bm{\kappa}') \nonumber\\
  &= \eta(\bm{\kappa}';\bm{x}) - \inf_{\bm{x}' \in \Omega} \left\{ \eta(\bm{\kappa}';\bm{x}')-\s(\bm{x}') \right\} \nonumber\\
  &= \eta(\bm{\kappa}';\bm{x}) - \eta(\bm{\kappa}'; \bm{x}_\mathrm{max}^\eta(\bm{\kappa}')) + \s(\bm{x}_\mathrm{max}^\eta(\bm{\kappa}'))
\end{align}
for all $\bm{\kappa}' \in K$.
The condition~\ref{cond:D} implies that there exists $\bm{\kappa}' \in K$
such that $\bm{x}_\mathrm{max}^\eta(\bm{\kappa}')=\bm{x}$,
then it holds
\begin{align}
  \inf_{\bm{\kappa} \in K} \left\{ \eta(\bm{\kappa};\bm{x})-\psi^\eta(\bm{\kappa}) \right\} &\leq \s(\bm{x}).
\end{align}
Therefore, we have Eq.~\eqref{eq:psi-s}.

\section{information-theoretic interpretation of Eq.~\eqref{eq:tdforce-formula}} \label{sec:derivation_tdforce}
In this section,
we give another derivation
of the formula for the intensive parameter,
Eq.~\eqref{eq:tdforce-formula}.
A great advantage of this proof is
that it works even in the case
where $\hat{\bm{x}}_N$ do not commute with each other
because it does not rely on the asymptotic approximation.

Suppose that the SE is outside the phase transition region,
and the {\CE}
\begin{align}
  \indexCE(\bm{\pi};\bm{x}) = \bm{\pi}\cdot\bm{x}
\end{align}
describes the same equilibrium state as the SE
in the thermodynamic limit
when $\bm{\pi}=\tilde{\bm{\pi}}$.
In particular, the following holds:
\begin{align}
  \lim_{N\to\infty} \bm{x}_N^\eta
  = \lim_{N\to\infty} \bm{x}_N^\indexCE(\tilde{\bm{\pi}}). \label{eq:tdforce-woLaplace_0}
\end{align}
Since $\displaystyle \lim_{N\to\infty} \bm{x}_N^\indexCE(\tilde{\bm{\pi}})$
is assumed to be outside the phase transition region,
the canonical thermodynamic function $\psi^\indexCE$ is differentiable at $\tilde{\bm{\pi}}$,
and its derivative at $\tilde{\bm{\pi}}$ is given by \cite{Griffiths1964}
\begin{align}
  \frac{\partial \psi^\indexCE}{\partial \pi_{i}}(\tilde{\bm{\pi}})
  =  \lim_{N\to\infty} \bm{x}_N^\indexCE(\tilde{\bm{\pi}}).
\end{align}
Therefore, using Eq.~\eqref{eq:tdforce-woLaplace_0},
Eq.~\eqref{eq:tdforce-formula} can be rephrased as
\begin{align}
  \lim_{N\to\infty} \left( \tilde{\pi}_{\ind{j}} - \frac{\partial\eta}{\partial x_{\ind{j}}}\left(\bm{x}_N^\indexCE(\tilde{\bm{\pi}})\right) \right) = 0. \label{eq:tdforce-woLaplace}
\end{align}
We derive Eq.~\eqref{eq:tdforce-woLaplace}
under reasonable conditions.

\subsection{Assumptions}
Let us measure the distance between the density matrices by the relative entropy
\begin{align}
  \infdiv*{\hat{\rho}}{\hat{\sigma}}
  \equiv \Tr \left[ \hat{\rho} \left\{ \log \hat{\rho} - \log \hat{\sigma} \right\} \right].
\end{align}
Since the ``distance'' between $\hat{\rho}_N^\eta$ and $\hat{\rho}_N^\indexCE(\bm{\pi})$ is expected to by minimized at $\bm{\pi}=\tilde{\bm{\pi}}$,
it is reasonable to assume that
\begin{enumerate}[label={(\roman*)}]
  \item \hfill\vspace{-\abovedisplayskip}\vspace{-\baselineskip}
    \begin{align}
    \left.\frac{\partial\infdiv*{\hat{\rho}_N^\indexCE(\bm{\pi})}{\hat{\rho}_N^\eta}}{\partial\pi_{\ind{i}}}\right|_{\bm{\pi}=\tilde{\bm{\pi}}} = o(N). \label{eq:tdforce-woLaplace_1}
  \end{align}
  \label{cond:tdforce-woLaplace_1}
\end{enumerate}
Note that a similar condition
\begin{align}
  \left.\frac{\partial\infdiv*{\hat{\rho}_N^\eta}{\hat{\rho}_N^\indexCE(\bm{\pi})}}{\partial\pi_{\ind{i}}}\right|_{\bm{\pi}=\tilde{\bm{\pi}}} = o(N)
\end{align}
is equivalent to Eq.~\eqref{eq:tdforce-woLaplace_0},
which is a necessary condition for $\hat{\rho}_N^\indexCE(\tilde{\bm{\pi}})$ and $\hat{\rho}_N^\eta$ to describe the same equilibrium state.
Therefore, it seems reasonable to assume Eq.~\eqref{eq:tdforce-woLaplace_1}.

In addition, we assume
\begin{enumerate}[label={(\roman*)}]
  \setcounter{enumi}{1}
  \item $\hat{\rho}_N^\indexCE(\tilde{\bm{\pi}})$ has normal fluctuations \cite{Jaksic2010,Goderis1989}:
    \begin{align}
      \braket{\delta\hat{x}_{N(i_1)} \cdots \delta\hat{x}_{N(i_{2n})}}_N^\indexCE(\tilde{\bm{\pi}})
      &= O(N^{-n}),
    \end{align}
    \label{cond:tdforce-woLaplace_2}
\end{enumerate}
where $\braket{\bullet}_N^\eta \equiv \Tr \left[ \bullet \rho_N^\eta \right]$.
Here, for simplicity of notation,
we abbreviate
$\braket{\cdots (\bullet-\braket{\bullet}_N^\eta) \cdots}_N^\eta$
to $\braket{\cdots \delta\bullet \cdots}_N^\eta$.

Finally, we assume
\begin{enumerate}[label={(\roman*)}]
  \setcounter{enumi}{2}
  \item the $m \times m$ positive-semidefinite matrix
    \begin{align}
      \Delta_{ij} \equiv \lim_{N\to\infty} N \Braket{\delta\hat{x}_{\indind{N}{i}};\delta\hat{x}_{\indind{N}{j}}}_N^\indexCE(\tilde{\bm{\pi}}) \label{eq:tdforce-woLaplace_3}
    \end{align}
    is invertible,
    \label{cond:tdforce-woLaplace_3}
\end{enumerate}
where $\braket{\bullet;\bullet}_N^\eta$ denotes
a generalization of the Bogoliubov-Duhamel inner product \cite{Petz1993}:
\begin{align}
  \braket{\hat{A};\hat{B}}_N^\eta
  &\equiv \int_0^1 d\nu\
  \braket{
    e^{+ \nu N \eta(\hat{\bm{x}}_N)} \hat{A}^\dagger
    e^{- \nu N \eta(\hat{\bm{x}}_N)} \hat{B}
}_N^\eta.
\end{align}
The validity of this assumption is explained as follows.
Since $\bm{x}^\indexCE(\tilde{\bm{\pi}})$ is assumed to be outside the phase transition region,
$\bm{x}^\indexCE(\bm{\pi})$ is analytic and one-to-one in a neighborhood of $\tilde{\bm{\pi}}$.
Therefore, $\displaystyle \frac{\partial x_{\ind{j}}^\indexCE}{\partial \pi_{\ind{i}}}(\tilde{\bm{\pi}})$ is an invertible matrix.
On the other hand, we have
\begin{align}
  \Delta_{ij} = \lim_{N\to\infty} \frac{\partial x_{\indind{N}{j}}^\indexCE}{\partial \pi_{\ind{i}}}(\tilde{\bm{\pi}}).
\end{align}
It is expected that we can interchange the limit and differentiation except at phase transition points (see Lemma~\ref{lemma:AA_cor}) as
\begin{align}
  \Delta_{ij} = \lim_{N\to\infty} \frac{\partial x_{\indind{N}{j}}^\indexCE}{\partial \pi_{\ind{i}}}(\tilde{\bm{\pi}})
  = \frac{\partial x_{\ind{j}}^\indexCE}{\partial \pi_{\ind{i}}}(\tilde{\bm{\pi}}).
\end{align}
Therefore, $\Delta_{ij}$ should be invertible.

Note that assumptions~\ref{cond:tdforce-woLaplace_2}-\ref{cond:tdforce-woLaplace_3}
are assumptions on the {\CE}, not on the SE.

\subsection{Proof}
To prove Eq.~\eqref{eq:tdforce-woLaplace},
we first prove the following lemma.

\begin{lemma} \label{lemma:bound_Bogoliubov-Duhamel}
If
\begin{align}
  \braket{\delta\hat{x}_{N(i_1)} \cdots \delta\hat{x}_{N(i_{2n})}}_N^\eta = O(N^{-n}),
\end{align}
then
\begin{align}
  \braket{
    \delta\hat{x}_{N(i_1)} \cdots \delta\hat{x}_{N(i_m)};
    \delta\hat{x}_{N(j_1)} \cdots \delta\hat{x}_{N(j_n)}
  }_N^\eta
  = O(N^{-\frac{m+n}{2}}). \label{eq:bound_Bogoliubov-Duhamel}
\end{align}
\end{lemma}
\begin{proof}
Applying the Cauchy-Schwarz inequality
\begin{align}
  \left| \braket{\hat{A};\hat{B}}_N^\eta \right|^2
  &\leq \braket{\hat{A};\hat{A}}_N^\eta \braket{\hat{B};\hat{B}}_N^\eta,
\end{align}
and Brooks Harris inequality \cite{Mermin1966,BrooksHarris1967}
\begin{align}
  \braket{\hat{A};\hat{A}}_N^\eta
  \leq \frac{1}{2} \braket{\{\hat{A}^\dagger,\hat{A}\}}_N^\eta,
\end{align}
we have
\begin{align}
  &\left| \braket{
    \delta\hat{x}_{N(i_1)} \cdots \delta\hat{x}_{N(i_m)};
    \delta\hat{x}_{N(j_1)} \cdots \delta\hat{x}_{N(j_n)}
  }_N^\eta \right|^2 \nonumber\\
  &\leq \frac{1}{4} \Braket{ \left\{
    \delta\hat{x}_{N(i_m)} \cdots \delta\hat{x}_{N(i_1)},
    \delta\hat{x}_{N(i_1)} \cdots \delta\hat{x}_{N(i_m)}
  \right\} }_N^\eta
  \Braket{ \left\{
    \delta\hat{x}_{N(j_n)} \cdots \delta\hat{x}_{N(j_1)},
    \delta\hat{x}_{N(j_1)} \cdots \delta\hat{x}_{N(j_n)}
  \right\} }_N^\eta.
\end{align}
Therefore, by assumption,
we obtain Eq.~\eqref{eq:bound_Bogoliubov-Duhamel}.
\end{proof}

\begin{proof}[Proof of Eq.~\eqref{eq:tdforce-woLaplace}]
By assumption~\ref{cond:tdforce-woLaplace_2}
and Lemma~\ref{lemma:bound_Bogoliubov-Duhamel},
we have
\begin{align}
  \braket{
    \delta\hat{x}_{\indind{N}{i}};
    \delta\hat{x}_{N(j_1)} \cdots \delta\hat{x}_{N(j_n)}
  }_N^\indexCE
  = O(N^{-\frac{1+n}{2}}).
\end{align}
In addition, since $\eta$ is a polynomial,
expanding in power series around $\bm{x}_N^\indexCE(\bm{\pi})$,
$\eta(\hat{\bm{x}}_N)$ can be written as a finite sum as
\begin{align}
  \eta(\hat{\bm{x}}_N)
  &= \eta(\bm{x}_N^\indexCE(\bm{\pi}))
  + \sum_j \frac{\partial\eta}{\partial x_{\ind{j}}}(\bm{x}_N^\indexCE(\bm{\pi})) \delta\hat{x}_{\indind{N}{j}}
  + \cdots.
\end{align}
Therefore, using
\begin{align}
  \frac{\partial\infdiv*{\hat{\rho}_N^\indexCE(\bm{\pi})}{\hat{\rho}_N^\eta}}{\partial\pi_{\ind{i}}}
  &= N^2 \Braket{
    \delta \left[ \partial_{\pi_{\ind{i}}}\indexCE(\bm{\pi};\hat{\bm{x}}_N) \right];
    \indexCE(\bm{\pi};\hat{\bm{x}}_N) - \eta(\hat{\bm{x}}_N)
  }_N^\indexCE(\bm{\pi}),
\end{align}
we have
\begin{align}
  \frac{\partial\infdiv*{\hat{\rho}_N^\indexCE(\bm{\pi})}{\hat{\rho}_N^\eta}}{\partial\pi_{\ind{i}}}
  &= \sum_j N^2 \left\{ \Braket{
    \delta\hat{x}_{\indind{N}{i}};
    \delta\hat{x}_{\indind{N}{j}}
  }_N^\indexCE(\bm{\pi}) \left( \pi_{\ind{j}} - \frac{\partial\eta}{\partial x_{\ind{j}}}(\bm{x}_N^\indexCE(\bm{\pi})) \right)
  + O(N^{-\frac{1+2}{2}}) \right\}.
\end{align}
Therefore, assumption~\ref{cond:tdforce-woLaplace_1} can be rewritten as
\begin{align}
  \sum_j \Delta_{ij} \lim_{N\to\infty} \left( \tilde{\pi}_{\ind{j}} - \frac{\partial\eta}{\partial x_{\ind{j}}}\left(\bm{x}_N^\indexCE(\tilde{\bm{\pi}})\right) \right)
  = 0.
\end{align}
Hence, by assumption~\ref{cond:tdforce-woLaplace_3},
we obtain Eq.~\eqref{eq:tdforce-woLaplace}.
\end{proof}

\section{Numerical implementation of the SE associated with $\eta$ in Eq.~\eqref{eq:d2qIsing_eta}} \label{sec:d2qIsing_method}
We calculate statistical-mechanical quantities in the SE associated with $\eta$ in Eq.~\eqref{eq:d2qIsing_eta}
using the METTS algorithm.
To obtain the METTS,
one needs to apply $e^{-\frac{1}{2}N\eta(\bm{\kappa};\hat{\bm{x}}_N)}$
to a random product state.
Choosing $\eta$ as Eq.~\eqref{eq:d2qIsing_eta},
this can be easily done as follows.

Consider the general case
where $\hat{X}_{\indind{N}{0}}$ and $\hat{X}_{\indind{N}{1}}$ can be written
as the sum of local observables:
\begin{align}
  \hat{X}_{\indind{N}{0}} &= \sum_{i=1}^N \hat{h}_i,\\
  \hat{X}_{\indind{N}{1}} &= \sum_{i=1}^N \hat{m}_i,
\end{align}
where $\hat{h}_i$ and $\hat{m}_i$ are observables located around site $i$.
Then it follows that
\begin{align}
  - \frac{1}{2} N \eta(\bm{\kappa};\hat{\bm{x}}_N)
  = \frac{1}{\delta} \left\{
    \sum_{i=1}^N \delta \hat{\gamma}_i
    + N\nu\delta \log \left( \sum_{i=1}^N \hat{m}_i + N \right)
    + \sum_{i=1}^N \delta \hat{\gamma}_i
  \right\}
  + \mathrm{const.},
\end{align}
where $\hat{\gamma}_i \equiv - \frac{1}{4} (\kappa_{\ind{0}}\hat{h}_i + \kappa_{\ind{1}}\hat{m}_i)$.
Therefore, by using the Trotter-Suzuki formula \cite{Suzuki1976},
$e^{-\frac{1}{2}N\eta(\bm{\kappa};\hat{\bm{x}}_N)}$ is decomposed
into a product of local operators as
\begin{align}
  e^{-\frac{1}{2}N\eta(\bm{\kappa};\hat{\bm{x}}_N)}
  \propto \left\{
    e^{\delta \hat{\gamma}_1} \cdots e^{\delta \hat{\gamma}_N}
    \left( \sum_{i=1}^N \hat{m}_i + N \right)^{N\nu\delta}
    e^{\delta \hat{\gamma}_N} \cdots e^{\delta \hat{\gamma}_1}
  \right\}^{1/\delta}
  + O(\delta^2). \label{eq:METTS_Trotter}
\end{align}
Especially when $N\nu\delta$ and $1/\delta$ are both integers,
the right hand side of Eq.~\eqref{eq:METTS_Trotter} can be obtained
by simply multiplying the local operators repeatedly.
In particular, by choosing $\nu$ so that $N\nu$ is an integer with many divisors,
it is possible to extrapolate the statistical-mechanical quantities to $\delta \to 0$
while taking $N\nu\delta$ as an integer,
as shown in Fig.~\ref{fig:METTS_Trotter}.
Hence, the SE can be numerically constructed easily even for quantum systems.

\begin{figure}
  \centering
  \includegraphics[keepaspectratio, width=0.5\linewidth]{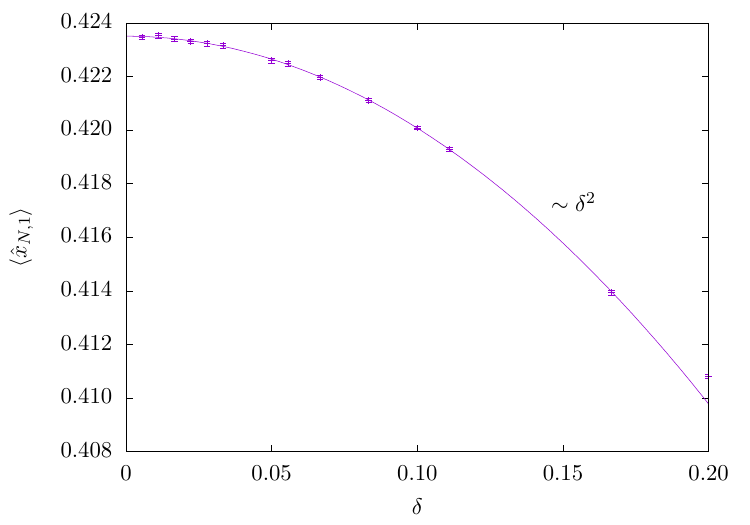}
  \caption{Trotter step dependence of the statistical-mechanical quantity
  for the one-dimensional transverse field Ising model
  ($\hat{X}_{\indind{N}{0}} = - J \sum \hat{\sigma}_{i}^z \hat{\sigma}_{i+1}^z - g \sum \hat{\sigma}_i^x,
  \hat{X}_{\indind{N}{1}} = \sum \hat{\sigma}_i^z$)
   in the SE for $N=50, J=1, g=1$ and $\nu=3.6, \bm{\kappa}=(1,5)$.}
  \label{fig:METTS_Trotter}
\end{figure}

\twocolumngrid
\bibliography{document}
\end{document}